\documentclass[12pt,reqno,final]{amsart}

\usepackage[color]{showkeys}     

\definecolor{refkey}{gray}{.5}   

\definecolor{labelkey}{gray}{.5} 

\usepackage{color}

\usepackage{epsfig}
\usepackage{amsmath, amsthm, amsfonts}
\usepackage{graphicx}
\usepackage{psfrag}

\numberwithin{equation}{section}

\newcommand{\R}{{\mathbb R}}
\newcommand{\C}{{\mathbb C}}
\newcommand{\N}{{\mathbb N}}

\newcommand{\Z}{{\mathbb Z}}

\newcommand{\Ai}{{\operatorname{Ai}}}

\newcommand{\al}{\alpha}

\newcommand{\Ga}{\Gamma}

\newcommand{\la}{\lambda}
\newcommand{\ep}{\varepsilon}
\newcommand{\de}{\delta}

\newcommand{\sg}{\sigma}

\newcommand{\z}{\zeta}

\newcommand{\hcal}{{\mathcal H}}

\newcommand{\mcal}{{\mathcal M}}
\newcommand{\acal}{{\mathcal A}}

\newcommand{\fcal}{{\mathcal F}}
\newcommand{\tcal}{{\mathcal T}}
\newcommand{\A}{{\mathbf A}}
\newcommand{\B}{{\mathbf B}}
\newcommand{\1}{{\mathbf 1}}
\newcommand{\Th}{{\Theta}}
\newtheorem{theo}{{\sc \bf Theorem}}[section]
\newtheorem{cor}[theo]{{\sc \bf Corollary}}

\newtheorem{prop}[theo]{{\sc \bf Proposition}}

\newenvironment{rem}{\medskip\noindent{\it Remark:\/} }{\medskip}

\newtheorem{lemma}{Lemma}[section]

\begin{document}

\title{Tail decay for the distribution of the endpoint of a directed polymer}
\author{Thomas Bothner}
\address{Department of Mathematical Sciences, Indiana University-Purdue University Indianapolis, 402 N. Blackford St., Indianapolis, IN 46202, U.S.A.}
\email{tbothner@iupui.edu}
\author{Karl Liechty}
\address{Department of Mathematics,
University of Michigan,
530 Church St., Ann Arbor, MI 48109, U.S.A.}
\email{kliechty@umich.edu}

\thanks{Part of this work was done during the workshop ``Vector equilibrium problems and their applications to random matrix models" hosted at the American Institute of Mathematics in April 2012. Both authors would like to thank the organizers for their hospitality.  KL would also like to thank Jinho Baik for helpful comments.}
\date{\today}

\begin{abstract}
We obtain an asymptotic expansion for the tails of the random variable $\tcal=\arg\max_{u\in\mathbb{R}}(\mathcal{A}_2(u)-u^2)$ where $\mathcal{A}_2$ is the Airy$_2$ process. Using the formula of Schehr \cite{Sch} that connects the density function of $\tcal$ to the Hastings-McLeod solution of the second Painlev\'e equation, we prove that as $t\rightarrow\infty$, $\mathbb{P}(|\tcal|>t)=Ce^{-\frac{4}{3}\varphi(t)}t^{-145/32}(1+O(t^{-3/4}))$, where $\varphi(t)=t^3-2t^{3/2}+3t^{3/4}$, and the constant $C$ is given explicitly.


\end{abstract}

\maketitle

\section{Introduction and statement of the main result}

Directed polymers in a random medium (DPRM) were introduced by Huse and Henley \cite{HH} to describe domain walls in a ferromagnetic Ising model with random impurities, sometimes called a dirty ferromagnet.  In the two dimensional Ising model, a typical domain wall is a lattice path in the plane, and DPRM 
 in 1+1 dimensions is a statistical mechanical model whose states are such paths.  For concreteness, we consider the square lattice $\N \times \N$, and the graph with edges connecting nearest neighbors,  so that the midpoints of the vertical edges have the Cartesian coordinates $(i, j+1/2)$ for some $i,j \in \N$, and the horizontal edges have the coordinates $(i+1/2,j)$ for some $i,j\in \N$.  On each vertex, we place a random weight $\ep_{ij}$ independently from some distribution, and consider some set of allowable lattice paths $\Ga$ which originate at $(0,0)$ and always move up and to the right.  We then define the random Gibbs measure on $\Ga$ as follows.  For a path $P\in \Ga$,
\begin{equation*}\label{in:1}
\mu(P)= \frac{1}{Z} \exp\left[\frac{1}{T} \sum_{(i,j) \in P} \ep_{ij}\right]\,,
\end{equation*}
where
\begin{equation*}\label{in:1a}
Z= \sum_{P\in \Ga} \exp\left[\frac{1}{T} \sum_{(i,j) \in P} \ep_{ij}\right]\,,
\end{equation*}
is the (random) partition function, and $T>0$ denotes temperature.  For the set of allowable paths $\Ga$, usually we consider 
\begin{equation*}\label{in:1b}
\Ga\equiv \Ga_{(m,n)} = \{ \textrm{up-right paths ending at the point } (m,n)\}\,,
\end{equation*}
in which case the model is said to have {\it point-to-point geometry}, or we consider
\begin{equation*}\label{in:1c}
\Ga\equiv \Ga_n = \{ \textrm{up-right paths of length } n \}\,,
\end{equation*}
in which case the model is said to have {\it point-to-line geometry}.  
\smallskip

In the limit as $T \to 0$, the Gibbs measure becomes a delta function on the path with the greatest weight, and the randomness in the model comes entirely from the random weights $\ep_{ij}$.  This is known as {\it directed last passage percolation}.    In the case of point-to-line directed last passage percolation with geometric or exponential weights on sites, it has been proven \cite{Joh} that the limiting fluctuations of both the energy of the maximizing path, and of the location of the endpoint of the polymer can be described in terms of the Airy$_2$ process.  
\smallskip

The Airy$_2$ process \cite{PS}, which we denote $\acal_2(u)$, is a stationary process whose marginal distributions are the Gaussian unitary ensemble (GUE) Tracy-Widom distribution \cite{TW}, and is expected to be a universal process governing the limiting spatial fluctuations of random growth models in the Kardar-Parisi-Zhang (KPZ) universality class \cite{KPZ} in $1+1$ dimension.  This has been proven in the case of the polynuclear growth model \cite{Joh, PS}. Now let $\mcal$ and $\tcal$ be the random variables defined by
\begin{equation*}\label{in:2}
	\mcal := \max_{u \in \mathbb{R}} \left( {\mathcal A}_2(u) - u^2 \right) \;,
\end{equation*}
and
\begin{equation*}\label{in:3}
	\tcal := \arg \max_{u \in \mathbb{R}} \left( {\mathcal A}_2(u) - u^2 \right) \;.
\end{equation*}
Then $\mcal$ describes the limiting fluctuations of the energy of the maximizing path in geometric point-to-line last passage percolation.  The limiting fluctuations of the endpoint of the path are described by $\tcal$.  This fact was proved by Johansson \cite{Joh} assuming that the maximum of $\acal_2(u)-u^2$ is attained at a unique point almost surely.  This assumption was later proved by Corwin and Hammond \cite{CH}.  For DPRM at finite temperature, there are some recent results.  For a continuum version and a semi-discrete version \cite{OY} of DPRM, as well as discrete DPRM with log-Gamma weights, the correct scaling exponents and a limit theorem for the free energy have been obtained at finite temperature \cite{BQS, ACQ, BCF, Sep, BC, CQ, SV}.  A limit theorem for the fluctuations of the endpoint of the polymer with point-to-line geometry, however, has only been proven in the case of geometric or exponential last passage percolation \cite{Joh}.  Nonetheless, $\tcal$ is expected to govern the fluctuations of the end-point of a polymer in DPRM at finite temperature as well, for a wide range of random weights $\ep_{ij}$.    For a review of the KPZ universality class in the physical literature, see \cite{HHZ}. For a more recent review in the mathematics literature, see \cite{Cor}.

\begin{rem}Although $\tcal$ has only been proven rigorously to govern endpoint fluctuations in the case of geometric last passage percolation, a formula for the endpoint fluctuations of a continuum version of DPRM was recently obtained nonrigorously by Dotsenko \cite{Dot}. This formula is equivalent to the known formulas for $\tcal$, given in \eqref{in:70} and \eqref{in:12} below.  The author of that paper was unable to prove this equivalence, and for the sake of completeness in the literature, we give a short proof in Appendix \ref{dot}.
\end{rem}
\smallskip

The distribution of $\mcal$ is, up to rescaling by a constant, the same as the Tracy-Widom distribution for the Gaussian orthogonal ensemble (GOE) \cite{TW2}.  Specifically, we have
\begin{equation*}
\mathbb{P}\left(\mcal \le t\right)=\fcal_1(2^{2/3} t)\,,
\end{equation*}
where $\fcal_1$ is the Tracy-Widom GOE distribution function, defined below in \eqref{TWA} and \eqref{TWP}.  This fact was proved by Johansson \cite{Joh} by first proving a functional limit theorem for the convergence of the polynuclear growth (PNG) model to the $\textrm{Airy}_2$ process and using connections between PNG and the longest increasing subsequence of a random permutation found by Baik and Rains \cite{BR}.    A more direct, although nonrigorous, proof was given in \cite{FMS} by analyzing the fluctuations of nonintersecting Brownian excursions.  A rigorous direct proof based on the explicit determinantal formula for the Airy$_2$ process was given in \cite{CQR}, and the approach of \cite{FMS} was made rigorous in \cite{Lie}. Since the Tracy-Widom GOE distribution has been well studied over the past 15-20 years, a lot is known about the distribution of $\mcal$.  In particular, the asymmetric tail behavior of $\fcal_1$ is
\begin{equation}\label{GOEtail}
\fcal_1(s)=\left\{
\begin{aligned}
1-\frac{e^{-\frac{2}{3}s^{3/2}}}{4\sqrt{\pi} s^{3/2}}\bigg(1+O\left(s^{-3/2}\right)\bigg)\,, \quad \textrm{as} \ s\to +\infty \\
\frac{\tau_1 e^{-\frac{|s|^3}{24}-\frac{|s|^{3/2}}{3\sqrt{2}}}}{|s|^{1/16}}\bigg(1+O\left(|s|^{-3/2}\right)\bigg)\,,
\quad \textrm{as} \ s\to -\infty\,,
\end{aligned}\right.
\end{equation}
where
\begin{equation*}
\tau_1=\frac{e^{\frac{1}{2}\z'(-1)}}{2^{11/48}},
\end{equation*}
and $\z(\cdot)$ denotes the Riemann zeta-function.
Similar formulas exist for the Tracy-Widom GUE and Gaussian symplectic ensemble (GSE) distribution functions, see \cite{BBD, BBDI, DIK}.  Recently, some similar formulas have appeared for the general $\beta$ Tracy-Widom distributions as well \cite{BEMN, BN, DV}. 

\smallskip

Much less attention has been dedicated to the study of $\tcal$. Exact expressions for the joint distribution of $(\mcal, \tcal)$ were obtained in two recent papers: in
\cite{M-FQR} by Moreno Flores, Quastel, and Remenik;
and  in \cite{Sch} by Schehr.  The formula of \cite{M-FQR} involves the Airy function and the resolvent of an associated operator and is derived rigorously, while the formula of \cite{Sch} involves a solution to the Lax pair for the Painlev\'e II equation, and is derived nonrigorously.  It was shown in \cite{BLS} that these formulas are indeed the same, and therefore the formula of \cite{Sch} is put on rigorous footing.
\smallskip

Let us describe the two formulas.  In order to do so, we first need to fix some notation which we will use throughout the paper.  Let $\hat {P}(m,t)$ denote the joint density function of $(\mcal, \tcal)$. 
Let $\Ai(x)$ be the Airy function \cite{BE}, and let $\B_s$ be the integral operator acting on $L^2[0,\infty)$ with kernel
\begin{equation}\label{in:4}
	\B_s(x,y)=\Ai(x+y+s),
\end{equation}
where $s\in \R$ is a parameter.  It is known that $\1-\B_s$ is invertible for any $s\in \R$. Let
\begin{equation}\label{in:5}
	\rho_s(x,y)=(\1-\B_s)^{-1} (x,y),  \qquad x,y\ge 0 \,.
\end{equation}
Define, for $t, m\in \R$, 
\begin{equation*}
	\psi(x; t, m) = 2 e^{x\, t} [t {\rm Ai}(t^2+m+x) + {\rm Ai}'(t^2+m+x)].
\end{equation*}
Then the formula of \cite{M-FQR} is 
\begin{equation}\label{in:70}
	\hat {P}(m,t) 
	= 2^{1/3} {\mathcal F}_1(2^{2/3} m) \int_0^{\infty}\!\!\int_{0}^\infty \psi(2^{1/3}x_1;-t,m) \rho_{2^{2/3}m}(x_1,x_2) \, \psi(2^{1/3} x_2;t,m)dx_1dx_2.
\end{equation}
In terms of the operator $\B_s$, the GOE Tracy-Widom distribution function $\fcal_1$ equals
\begin{equation}\label{TWA}
{\mathcal F}_1(s)= \det(\1-\B_s),
\end{equation} see \cite{FS, Sas}.
\smallskip

We now present the formula of \cite{Sch}. To this end let $q(s)$ be the particular solution of the second Painlev\'{e} equation,
\begin{equation*}
	q''(s)=sq(s)+2q(s)^3\,,
\end{equation*}
satisfying 
\begin{equation*}
	q(s) \sim \Ai(s)\,, \quad \textrm{as} \ s \to + \infty \;.
\end{equation*}
This particular solution is known as the Hastings-McLeod solution \cite{HastingsMcLeod, Itsbook},  
and its uniqueness and global existence are well established. Consider now the following Lax-system (cf. \cite{BI}) associated  
to the Hastings-McLeod solution of the Painlev\'e II equation, i.e. the system of linear differential equations for a two-dimensional vector 
$\Phi=\Phi(\zeta,s)$, 
\begin{equation}\label{in:6}
 \frac{\partial\Phi}{\partial \zeta} = A \Phi \;, \quad \; \frac{\partial\Phi}{\partial s} = B \Phi \;,
 \end{equation}
 where the $2 \times 2$ matrices $A = A(\zeta, s)$ and $B = B(\zeta,s)$ are given by
 \begin{equation}\label{in:7}
 A(\zeta,s) = \left( 
 \begin{array}{c c}
 4 \zeta q &  4 \zeta^2 + s + 2q^2 + 2q'\\
 -4 \zeta^2 - s - 2 q^2 + 2q' & -4 \zeta q
 \end{array}\right) \;,
\end{equation}
and 
\begin{equation}\label{in:8}
  B(\zeta,s) = \left( 
 \begin{array}{c c}
 q &  \zeta \\
 -\zeta & - q
 \end{array}\right) \;.
\end{equation}
The above system \eqref{in:6} is overdetermined, and the compatibility of the equations implies that $q(s)$ solves the Painlev\'e II equation. 
Now let $\Phi =\binom{\Phi_1}{\Phi_2}$ be the unique solution of (\ref{in:6}) which satisfies  
the real asymptotics
\begin{equation*}\label{in:9}
	\Phi_1(\z;s)=\cos\left(\frac{4}{3}\z^3+s\z\right)+O(\z^{-1}), \quad \Phi_2(\z;s)=-\sin\left(\frac{4}{3}\z^3+s\z\right)+O(\z^{-1})\,,
\end{equation*}
as $\z \to \pm \infty$ for $s\in \R$. Such a solution exists (see e.g. \cite{BI, DZ, Itsbook}), and it  further satisfies the property that $\Phi_1(\z;s)$ and $\Phi_2(\z;s)$ are real for real $\z$ and $s$, as well as the symmetry relations
\begin{equation*}
	\Phi_1(-\z;s)=\Phi_1(\z;s), \hspace{1cm} \Phi_2(-\z;s)=-\Phi_2(\z;s),\hspace{0.5cm}\z\in\mathbb{C},s\in\mathbb{R}.
\end{equation*}
Define for non-negative $w$ and real $s$,
\begin{equation}\label{in:11}
	h(s,w) :=  \int_{0}^\infty  \zeta \Phi_2(\zeta;s) e^{- w \zeta^2} \, d\zeta\,.
\end{equation}
The formula of \cite{Sch} for the joint density of $\mcal$ and $\tcal$ is
\begin{equation}\label{in:12}
\begin{aligned}
	{\hat P}(m,t) &= 4 P(2^{2/3}m, 2^{4/3} t) \;, 
\end{aligned}
\end{equation}
where
\begin{equation}\label{in:13}
	P(s,w) = \frac{4}{\pi^2} {\mathcal F}_1(s) \int_{s}^\infty h(u,w) h(u,-w) \, du,
\end{equation}
and $h(s,-w)$, $w>0$ is understood as the analytic continuation of $h(\cdot,w)$ to the negative real axis.
In terms of the Painlev\'{e} function $q(s)$, the Tracy-Widom GOE distribution function $\fcal_1$ can be written as \cite{TW2}
\begin{equation}\label{TWP}
\fcal_1(s)=\exp\left[-\frac{1}{2}\int_s^\infty q(x)dx-\frac{1}{2} \int_s^\infty\int_t^\infty q(x)^2\,dx\,dt\right]\,.
\end{equation}
 The marginal density function for $\tcal$ equals
\begin{equation*}\label{in:14}
	{\hat P}(t):=\int_{-\infty}^\infty {\hat P}(m,t)\,dm\,,
\end{equation*}
and using \eqref{in:12}, we have
\begin{equation}\label{in:14f}
		{\hat P}(t):=2^{4/3}P(2^{4/3} t)\,,
\end{equation}
where we introduced
\begin{equation}\label{in:14b}
		P(w)=\int_{-\infty}^\infty P(s, w)\,ds\,.
\end{equation}
Our main result in the present paper is an asymptotic expansion of the density function  ${\hat P}(t)$ for large $t$.  

\begin{theo}\label{main} As $t\to \infty$, the marginal density function  ${\hat P}(t)$ satisfies
\begin{equation*}\label{m:1}
 {\hat P}(t)= \tau e^{-\frac{4}{3}\varphi(t)}t^{-81/32}\left(1+\frac{15}{4t^{3/4}}+O(t^{-3/2})\right)\,,
 \end{equation*}
which extends to a full asymptotic series in powers of $t^{-3/4}$.  The function $\varphi(t)$ is
\begin{equation*}
\varphi(t)=t^3-2t^{3/2}+3t^{3/4},
\end{equation*}
and the constant $\tau$ is given by
\begin{equation*}\label{m:2}
	\tau= 2^{-29/6}e^{5/4}e^{\frac{1}{2}\z'(-1)} \pi^{3/2},
\end{equation*}
with $\zeta(z)$ denoting the Riemann zeta-function.
\end{theo}
Integrating the above density, Theorem \ref{main} implies the following corollary.
\begin{cor}\label{main2}
As $t \to \infty$, 
\begin{equation}\label{m:3}
	\mathbb{P}(|\tcal|> t)= Ce^{-\frac{4}{3}\varphi(t)}t^{-145/32}\left(1+\frac{15}{4t^{3/4}}+O(t^{-3/2})\right)\,.
\end{equation}
where $C=\tau/2$.
\end{cor}

\begin{rem}
The leading decay order of $e^{-ct^3}$ in \eqref{m:3} was first predicted in the physics literature by Halpin-Healy and Zhang \cite{HHZ}, compare also \cite{MP}.  The first rigorous confirmation of this rate of decay appeared in the paper \cite{CH} of Corwin and Hammond, in which they give $e^{-ct^3}$ as an upper bound on $\mathbb{P}(|\tcal|> t)$, although they do not give the value of the constant $c$.   In the paper \cite{Sch} of Schehr, the author uses the formula \eqref{in:12} to find that the leading coefficient is $c=4/3$, although he did not employ complete and rigorous estimates. Recently Quastel and Remenik \cite{QR} rigorously obtained explicit bounds on $c$, and remarked that they believe that  the correct rate of decay is in fact $4/3$.  Corollary \ref{main2} confirms this rigorously, and in addition gives subleading terms and constants.  In principle, all terms in the asymptotic expansion are computable by the methods of this paper, but calculations become more involved.
\end{rem}

The setup for the remainder of the paper is a follows. We will prove Theorem \ref{main} using \eqref{in:13}, \eqref{in:14f} and \eqref{in:14b}.  In Section \ref{plusinf} we evaluate $h(s,w)$ as $w\to\infty$ by analyzing the system \eqref{in:6} close to the origin.  In Section \ref{minusinf} we prepare $h(s,-w)$ for asymptotic analysis as $w\rightarrow \infty$ using the global identity from \cite{BLS} which allows us to express $\Phi_2(\z;s)$ in terms of the resolvent of a Hankel operator on $L^2[0,\infty)$ whose kernel is constructed out of the Airy function. Then in Section \ref{largeexp} we split $P(w)$ into two parts, one which can be estimated using asymptotics \eqref{GOEtail} of the Tracy-Widom distribution function $\fcal_1(s)$, and another which we evaluate asymptotically by Laplace's method. Theorem \ref{main} then follows via \eqref{in:14f}, and the proof of Corollary \ref{main2} will be given in Section \ref{proofcor}.

\medskip


\section{Expansion of $h(s,w)$ as $w\to +\infty$}\label{plusinf}
With the change of variables $\la=\z \sqrt{2w}$, (\ref{in:11}) becomes
\begin{equation}\label{h:1}
h(s,w)=\frac{1}{2w} \int_0^\infty \la \Phi_2 \left(\frac{\la}{\sqrt{2w}}; s\right) e^{-\frac{\la^2}{2}} d\la\,.
\end{equation}
We thus see that we need the Taylor expansion of  $\Phi_2(\z;s)$ at $\z =0$. To this end let us consider \eqref{in:6} as a system for a $2\times 2$ matrix-valued function $\Psi$,
\begin{equation}\label{h:2}
	\frac{\partial \Psi}{\partial \z} = \bigg(\sum_{n=0}^2A_n(s)\z^n\bigg)\Psi,\ \ \frac{\partial \Psi}{\partial \z} = \bigg(\sum_{n=0}^1B_n(s)\z^n\bigg)\Psi,\ \ \  \Psi(\zeta,s)=\big(\Phi(\zeta,s),\hat{\Phi}(\zeta,s)\big)\,,
\end{equation}
where
\begin{equation*}\label{h:2a}
\begin{aligned}
	A_0=\begin{pmatrix} 
		0 & v(s)+4q'(s) \\
 		-v(s) & 0 
 			\end{pmatrix}\,, \qquad A_1&=\begin{pmatrix}
 			  4q(s) & 0 \\ 
 			  0 & -4q(s) 
 			  \end{pmatrix}\,, \qquad A_2&=\begin{pmatrix}
 			   0 & 4 \\
 			   -4 & 0 
 			   \end{pmatrix}\,, \\
B_0&=\begin{pmatrix}  
q(s) & 0 \\
 0 & -q(s)
  \end{pmatrix}\,, \qquad \hspace{0.45cm}B_1&=\begin{pmatrix}  0 & 1 \\ -1 & 0 \end{pmatrix}\,,
\end{aligned}
\end{equation*}
are determined from \eqref{in:7} and \eqref{in:8}, and
\begin{equation}\label{h:2b}
v(s):=s+2q(s)^2-2q'(s)\,.
\end{equation}
The matrix function $\Psi$ can alternatively be defined as the solution of a certain oscillatory Riemann-Hilbert problem, see Appendix \ref{app1}. Such Riemann-Hilbert formulation appears in one form or another throughout the literature on the integrable structure of Painlev\'{e} II, and is the basis for many results (cf. \cite{Itsbook}).  In particular, we will use the facts that
\begin{equation}\label{rhp-5}
	\Psi(-\z;s)=\sg_3 \Psi(\z;s) \sg_3,\hspace{0.5cm}\z\in\mathbb{C},\ s\in\mathbb{R},
\end{equation}
and
\begin{equation}\label{rhp-6}
	\Psi(0;s) = e^{-\sg_3 \int_s^\infty q(t)\,dt}\,,
\end{equation}
which are justified in Appendix \ref{app1}.
\smallskip

Since the coefficient function $A(\z,s)$ is analytic at the origin, there exists a solution $\Psi(\z,s)$ of equation \eqref{h:2} which is holomorphic in some neighborhood of the origin and uniquely determined by its value $\Psi(0,s) = \Psi_0(s)$ (cf. \cite{W}).  The Taylor expansion of $\Psi$ about $\z=0$ is
\begin{equation}\label{h:3}
	\Psi(\z,s) = \sum_{n=0}^{\infty}\Psi_n(s)\zeta^n,\ \ |\z|<\rho,\ \ \rho>0.
\end{equation}
Inserting \eqref{h:3} into \eqref{h:2} and comparing coefficients yields for $n\in\Z_{\geq 0}$ first
\begin{equation}\label{h:4}
	(n+1)\Psi_{n+1}=A_0\Psi_n(s)+A_1\Psi_{n-1}(s)+A_2\Psi_{n-2}(s),\ \ \ \Psi_{-1}(s)=\Psi_{-2}(s)\equiv 0,
\end{equation}
and secondly
\begin{equation*}\label{h:4a}
	\frac{d\Psi_n}{ds} = q(s)\sigma_3\Psi_n(s)+i\sigma_2\Psi_{n-1}(s),\ \ \ \sigma_3=\begin{pmatrix}
	1 & 0\\
	0 & -1\\
	\end{pmatrix},\ \sigma_2=\begin{pmatrix}0 & -i\\
	i & 0\\
	\end{pmatrix}.
\end{equation*}
The latter differential recursion implies
\begin{equation*}\label{h:5}
	\Psi_0(s) = e^{-\sigma_3\int_s^{\infty}q(x)dx}C_0,
\end{equation*}
for some $2\times 2$ invertible matrix valued constant $C_0$.  In fact, from (\ref{rhp-6}) we see that $C_0= I$.
In order to obtain the rest of the Taylor coefficients, we use the recursion \eqref{h:4} to write $\Psi_k$ explicitly in terms of the matrices $A_0, A_1, A_2$.  The first two coefficients are
\begin{equation}\label{h:6}
	\Psi_1(s) = A_0e^{-\sigma_3\int_s^{\infty}q(x)dx},\ \ \Psi_2(s)=\frac{1}{2}\big(A_0^2+A_1\big)e^{-\sigma_3\int_s^{\infty}q(x)dx}.
\end{equation}
At this point we would like to connect the coefficients $\Psi_n(s)$ to the required values of $\Phi_1(\z;s)$ and $\Phi_2(\z;s)$.
The aformentioned global symmetry relation \eqref{rhp-5} implies
\begin{equation*}\label{h:8}
	\Psi_n(s) = (-1)^n\sigma_3\Psi_n(s)\sigma_3,\ n\in\Z_{\geq 0},
\end{equation*}
i.e. all coefficients $\Psi_{2n}(s)$ are diagonal whereas all coefficients $\Psi_{2n+1}(s)$ are off-diagonal.  We therefore have
\begin{equation}\label{h:9}
	\Phi_1(\z;s) = \sum_{n=0}^{\infty}\Psi_{2n}^{(11)}(s)\z^{2n},\ \ \ \Phi_2(\z;s)=\sum_{n=0}^{\infty}\Psi_{2n+1}^{(21)}(s)\z^{2n+1},\ \ \ \Psi_n=\big(\Psi_n^{(ij)}\big)_{i,j=1}^2,\ \ |\z|<\rho.
\end{equation}
From \eqref{h:4} (see explicitly \eqref{h:6}),
\begin{equation}\label{h:10}
	\Psi_{2n+1}^{(21)}(s) = P_n\big(s,q(s),q'(s)\big)e^{-\int_s^{\infty}q(x)dx},
\end{equation}
with polynomials $P_n$ in three variables. Integrating \eqref{h:1}, we obtain the following Proposition.
\begin{prop}\label{prop0} As $w\to +\infty$, the function $h(s,w)$ satisfies the asymptotic expansion
\begin{equation}\label{h:11a}
		h(s,w) =\frac{\sqrt{\pi}}{4w^{3/2}} e^{-\int_s^\infty q(x)\,dx} \sum_{n=0}^\infty Q_n(s,q(s),q'(s))w^{-n}\,,
\end{equation}
where
\begin{equation*}\label{h:11b}
		Q_n(s,q(s),q'(s))=e^{\int_s^\infty q(x)\,dx} \Psi_{2n+1}^{(21)}(s) \frac{(2n+1)!}{n! 4^{n}},
\end{equation*} 
is a polynomial in 3 variables.  This expansion is uniform in $s$ on compact subsets of the real line.  The first two terms in the expansion \eqref{h:11a} are
\begin{equation}\label{Q01}
		Q_0(s,q(s),q'(s))=-v(s)\,, \qquad Q_1(s,q(s),q'(s))=\frac{v(s)^3}{4}+v(s)\big(v(s)q'(s)+q(s)\big)-2\,,
\end{equation}
where $v(s)$ is defined in \eqref{h:2b}.
Moreover, as $s\to -\infty$, the polynomial $Q_n(s,q(s), q'(s))$ satisfies
\begin{equation}\label{h:11c}
		Q_n(s,q(s),q'(s))=O\left((-s)^{(n-1)/2}\right).
\end{equation} 
\end{prop}
\begin{proof}
Replacing $\Phi_2$ in (\ref{h:1}) by its power series \eqref{h:9} gives
\begin{equation*}
\begin{aligned}
	h(s,w)&=\frac{1}{(2w)^{3/2}} \sum_{n=0}^\infty \frac{\Psi_{2n+1}^{(21)}(s)}{(2w)^{n}} \int_0^\infty \la^{2n+2}e^{-\frac{\la^2}{2}} \,d\la \\
	&=\frac{\sqrt{\pi}}{4w^{3/2}} \sum_{n=0}^\infty \frac{\Psi_{2n+1}^{(21)}(s) (2n+1)!}{4^nn!}w^{-n}\,.
\end{aligned}
\end{equation*}
Using the structure \eqref{h:10} of $\Psi_{2n+1}^{(21)}(s)$, we can thus write, as $w\to\infty$,
\begin{equation*}
h(s,w)=\frac{\sqrt{\pi}}{8w^{3/2}}e^{-\int_s^\infty q(x)\,dx} \sum_{n=0}^\infty Q_n(s,q(s),q'(s)) w^{-n}\,.
\end{equation*}
For the asymptotics \eqref{h:11c}, we use \cite{DZ}\footnote[2]{We would like to point out the following small detail. In \cite{BBD} the coefficient of $O(s^{-9})$ is given as $\frac{10219}{1024}$, whereas in \cite{DZ} and \cite{TW} it is $\frac{10657}{1024}$.}. Namely as $s\to -\infty$,
\begin{equation}\label{PIIasy} 
\begin{aligned}
	q(s)  &= \sqrt{\frac{-s}{2}}\left(1+\frac{1}{8s^3}-\frac{73}{128s^6}+\frac{10657}{1024s^9}+O(s^{-12})\right)\,,\\
 	q'(s) &= -\frac{1}{2^{3/2} \sqrt{-s}}\left(1-\frac{5}{8s^3}+\frac{803}{128s^6}-\frac{181169}{1024s^9}+O(s^{-12})\right)\,,
\end{aligned}
\end{equation}
from which it follows 
\begin{equation*}
\begin{aligned}
	v(s) &=& \frac{1}{\sqrt{-2s}}\bigg(1-\frac{(-s)^{-3/2}}{2\sqrt{2}}-\frac{5}{8s^3}-\frac{9(-s)^{-9/2}}{4\sqrt{2}}
	+\frac{803}{128s^6}-\frac{1323(-s)^{-15/2}}{32\sqrt{2}}\\
	&&\hspace{3cm} -\frac{181169}{1024s^9}+O\big((-s)^{-21/2}\big)\bigg).
\end{aligned}
\end{equation*}
From these we find that, as $s \to -\infty$,
\begin{equation*}
\begin{aligned}
	\Psi_{2n}(s)&=\begin{pmatrix} O((-s)^{n/2}) & 0 \\ 0 & O((-s)^{n/2})\end{pmatrix}e^{-\sg_3\int_s^\infty q(x)\,dx}\,, \\
	\Psi_{2n+1}(s)&=\begin{pmatrix} 0 & O((-s)^{(n-1)/2}) \\ O((-s)^{(n-1)/2}) & 0 \end{pmatrix}e^{-\sg_3\int_s^\infty q(x)\,dx}\,,
\end{aligned}
\end{equation*}
which is easily checked by induction using the recursion \eqref{h:4}.
\end{proof}

Let us now consider $h(s,w)$ as $w\rightarrow-\infty$.

\section{Airy formula for $h(s,w)$ and discussion}\label{minusinf}
In order to analyze $h(s,w)$ for large negative values of $w$ it is convenient to use a formula for $\Phi_2$ which was obtained in \cite{BLS}. Let us review the notation of \cite{BLS}.
As in \eqref{in:4}, let $\B_s$ be the integral operator acting on $L^2[0,\infty)$ with kernel
\begin{equation*}\label{pr:0}
	\B_s(x,y)=\Ai(x+y+s).
\end{equation*}
 Let $\A_s:=\B_s^2$ be the Airy operator acting on $L^2[0,\infty)$, which has kernel
\begin{equation*}\label{pr:1}
\begin{split}
	\A_s(x,y)
	&=\int_0^\infty \Ai(x+s+\xi) \Ai (y+s+\xi)\,d\xi \\
	&= \frac{\Ai(x+s)\Ai'(y+s)-\Ai'(x+s)\Ai(y+s)}{x-y} \,.
\end{split}
\end{equation*}
Define the functions $Q$ and $R$ as
\begin{equation}\label{pr:2}
	Q:= (\1 - \A_s)^{-1} \B_s \de_0\,, \qquad R:= (\1 - \A_s)^{-1} \A_s \de_0\,,
\end{equation}
where $\de_0$ is the Dirac delta function at zero. As in \cite{BLS} we use the convention that 
the Dirac delta function satisfies $\int_{[0, \infty)} \delta_0(x) f(x)dx =f(0)$ for functions $f$ which are right-continuous at $0$.
Introduce also the function
\begin{equation*}\label{pr:3}
	\Th(x) :=-\sin\left( \frac43 \zeta^3+(s+2x)\zeta\right)\,.
\end{equation*}
Proposition 2.1 of \cite{BLS} is
\begin{equation}\label{pr:5}
	\Phi_{2}(\z,s) = \Th(0)+\langle \Th, R + Q \rangle\,,
\end{equation}
where the functions $Q$ and $R$ are defined in (\ref{pr:2}), and $\langle \cdot, \cdot \rangle$ is the inner product on $L^2[0,\infty)$.
Notice that the function $R+Q$ can be written in terms of the resolvent kernel of the operator $\B_s$.  Indeed, let 
\begin{equation*}
\mathbf R_s:=(\1-\B_s)^{-1} -\1\,,
\end{equation*}
 be the resolvent of the operator $\B_s$.  It is an integral operator as well, and let us use $K_s$ to denote the function $\mathbf R_s \de_0$. We then have 
\begin{equation}\label{pr:6}
\begin{aligned}
R+Q&=(\1 - \B_s^2)^{-1}(\B_s+ \B_s^2)\de_0=(\1-\B_s)^{-1} \B_s \de_0=(\1-\B_s)^{-1}(\1-(\1- \B_s)) \de_0 \\
&=((\1-\B_s)^{-1}-\1) \de_0= K_s,
\end{aligned}
\end{equation}
and (\ref{pr:5}) becomes
\begin{equation}\label{pr:7}
	\Phi_{2}(\z,s) = \Th(0)+\langle \Th, K_s \rangle\,.
\end{equation}

Substituting \eqref{pr:7} into \eqref{in:11} gives
\begin{equation}\label{pr:8}
\begin{aligned}
	h(s,w) &=& -\int_0^{\infty}\z\sin\bigg(\frac{4}{3}\z^3+s\z\bigg)e^{-w\z^2}d\z+\int_0^{\infty}\z\left\langle\Th,K_s\right\rangle e^{-w\z^2} d\z \\
	&=&a(s,w)+\int_0^{\infty}a(s+2x,w)\Big((\1-\B_s)^{-1}\mathbf{T}_s\Ai\Big)(x)\,dx\,,
	\end{aligned}
\end{equation}
with
\begin{equation}\label{pr:9}
	a(y,w) = -\int_0^{\infty}\z\sin\bigg(\frac{4}{3}\z^3+y\z\bigg)e^{-w\z^2}d\z,\ \ y\in\mathbb{R},
\end{equation}
and the translation operator
\begin{equation*}\label{pr:10}
	\big(\mathbf{T}_sf\big)(x)=f(x+s).
\end{equation*}
When $w$ lies off of the non-negative real axis \eqref{pr:9} diverges. In order to obtain the analytical continuation of $a(\cdot,w)$ to the complex $w$-plane, set $\z=-iz$. Then
\begin{equation*}\label{pr:11}
	a(y,w) = \frac{1}{i}\int_0^{i\cdot\infty}z\sinh\bigg(\frac{4}{3}z^3-yz\bigg)e^{wz^2}dz = \frac{1}{2i}\int_{-i\cdot\infty}^{i\cdot\infty}ze^{\frac{4}{3}z^3+wz^2-yz}dz,
\end{equation*}
and by Cauchy's theorem we can deform the contour of integration such that
\begin{equation}\label{pr:12}
	a(y,w)=\frac{1}{2i}\int_{\mathcal{L}}ze^{\frac{4}{3}z^3+wz^2-yz}dz,
\end{equation}
where $\mathcal{L}$ is any contour that starts at the point infinity with argument $\textnormal{arg}\ z=-\frac{\pi}{3}$ (see Figure \ref{fig2}) and ends at the conjugated point.
\begin{figure}[tbh]
  \begin{center}
  \psfragscanon
  \psfrag{1}{\footnotesize{$\mathcal{L}$}}
  \includegraphics[width=8cm,height=6cm]{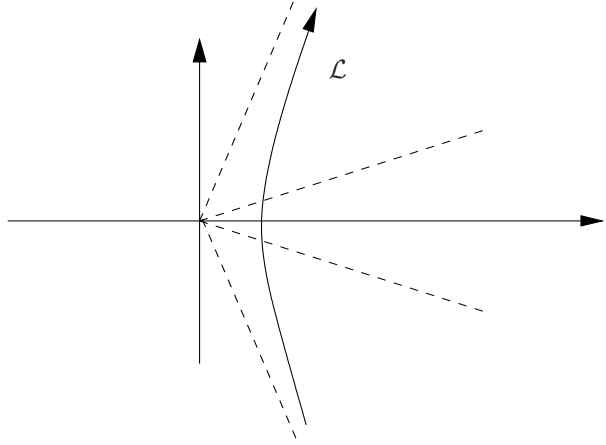}
  \end{center}
  \caption{The integration contour $\mathcal{L}$ chosen in \eqref{pr:12}. The dashed lines are the rays $\textnormal{arg}\ z=\pm\frac{\pi}{6}$ and $\textnormal{arg}\ z=\pm\frac{\pi}{3}$.}
  \label{fig2}
\end{figure}
We make the following observations. First \eqref{pr:12} converges absolutely and uniformly in any compact $w$ domain, hence \eqref{pr:12} determines the analytic continuation of $a(\cdot,w)$ to the entire real line and it allows us to consider \eqref{pr:8} now for all $w\in\mathbb{R}$. Secondly using the definition
\begin{equation*}\label{pr:13}
	\Ai(z) = \frac{1}{2\pi i}\int_{\mathcal{L}}e^{\frac{1}{3}t^3-tz}dt,\ \ z\in\mathbb{C},
\end{equation*}
we obtain (see \cite{BLS} for a similar computation),
\begin{equation*}\label{pr:14}
	a(y,w)=\frac{\pi}{2^{4/3}}e^{\frac{w^3}{24}+\frac{wy}{4}}\bigg(\Ai'\big(2^{-2/3}y+2^{-8/3}w^2\big)+\frac{w}{2^{4/3}}\Ai\big(2^{-2/3}y+2^{-8/3}w^2\big)\bigg),\ \ w,y\in\mathbb{R}.
\end{equation*}
In particular, notice that, for $y\gg -w^2/4$ we have, as $w\rightarrow\infty$,
\begin{eqnarray}
	a(y,-w)&=&-\frac{\sqrt{\pi}\sqrt{w}}{8}e^{-\frac{w^3}{24}-\frac{wy}{4}}\sum_{n,m=0}^{\infty}\bigg(\frac{4y}{w^2}\bigg)^n\bigg(-\frac{3}{2}\bigg)^m
	\Big(2^{-2/3}y+2^{-8/3}w^2\Big)^{-\frac{3m}{2}}\nonumber\\
	&&\times\Big[\binom{1/4}{n}d_m+\binom{-1/4}{n}c_m\Big]\exp\Bigg[-\frac{2}{3}\Big(2^{-2/3}y+2^{-8/3}w^2\Big)^{3/2}\Bigg]\,, \label{pr:15}
\end{eqnarray}
with (cf. \cite{BE})
\begin{equation*}
	c_0=d_0=1,\hspace{0.5cm} c_n = \frac{\Gamma(3n+\frac{1}{2})}{54^nn!\Gamma(n+\frac{1}{2})},\ \ d_n = -\frac{6n+1}{6n-1}c_n,\ \ n\geq 1.
\end{equation*}
We now go back to \eqref{pr:8} and investigate the kernel $K_s(x)$.  We consider it as a function of both $x$ and $s$, so let us write
\begin{equation*}\label{pr:16}
	K(x,s)\equiv K_s(x) =\Big((\1-\B_s)^{-1}\mathbf{T}_s\Ai\Big)(x),\ \ x\geq 0,
\end{equation*}
which solves the Fredholm integral equation
\begin{equation}\label{pr:17}
	\Ai(x+s)=K(x,s)-\int_0^{\infty}\Ai(x+y+s)K(y,s)dy.
\end{equation}
Recall at this point that the Hankel operator $\B_s:L^2[0,\infty)\rightarrow L^2[0,\infty)$ is a bounded operator and $\1-\B_s$ is invertible for any fixed $s\in\mathbb{R}$. Thus $K(x,\cdot)\in L^1[0,\infty)$.  In fact, as a consequence of \eqref{pr:17}, we have $K(x,\cdot)\in C^{\infty}[0,\infty)\cap L^1[0,\infty)$ for fixed $s\in\mathbb{R}$ and
\begin{equation*}\label{pr:18}
	K(x,s) = \Ai(x+s)\big(1+o(1)\big),\ x\rightarrow\infty,
\end{equation*}
uniformly on any compact subset of $\{s:-\infty<s<\infty\}$. For regularity of $K(s,\cdot)$ as a function in $s$, we get via H\"{o}lder's inequality and \eqref{pr:17},
\begin{equation*}\label{pr:19}
	K(s,\cdot)\in C^{\infty}(\mathbb{R}).
\end{equation*}
To obtain a statement on the large positive and negative $s$ behavior of $K(s,\cdot)$, we can view $\B_s$ equivalently as an operator $\B$ acting on $L^2[s,\infty)$ with kernel
\begin{equation*}
	\B(x,y) = \Ai(x+y).
\end{equation*}
A small norm argument for $\B$ then shows that $K(s,\cdot)$ is approaching zero exponentially fast as $s\rightarrow +\infty$ for any fixed $x\in\mathbb{R}$. On the other hand from \eqref{pr:17}, $K(s,\cdot)$ can grow at most power-like as $s\rightarrow-\infty$, again for any fixed $x\in\mathbb{R}$. For our purposes we are especially interested in the behavior of $K(0,s)$ as $s\rightarrow-\infty$. By \eqref{pr:6}, $K(0,s)$ is 
\begin{equation*}
	K(0,s) = R(0)+Q(0), 
	\end{equation*}
and thus (see \cite{TW, BLS}),
\begin{equation*}
	K(0,s) = \int_s^{\infty}q^2(t)\,dt+q(s).
\end{equation*}
Hence (cf. \cite{DZ},\cite{TW}),
\begin{equation}\label{pr:22}
	K(0,s) = \frac{s^2}{4}\bigg(1+\frac{2\sqrt{2}}{(-s)^{3/2}}-\frac{1}{2s^3}-\frac{\sqrt{2}}{4(-s)^{9/2}}+O\big(s^{-6}\big)\bigg),\hspace{0.5cm}s\rightarrow-\infty.
\end{equation}

\section{Large $w$ expansion of $P(w)$}\label{largeexp}
We have now gathered enough information to compute the asymptotics of $P(w)$ as $w\rightarrow\infty$. Let us write $P(w)$ as
\begin{equation*}\label{exp:1}
	P(w)=P_1(w)+P_2(w),
\end{equation*}
where
\begin{equation*}\label{exp:2}
\begin{aligned}
	P_1(w)&=\int_{-\infty}^\infty \, \fcal_1(s)\bigg[ \int_s^\infty \, h(u,w) a(u,-w)du\bigg] ds\,, \\
	P_2(w)&=\int_{-\infty}^\infty \, \fcal_1(s)\bigg[ \int_s^\infty \, h(u,w)\bigg\{\int_0^\infty \, a(u+2x,-w)K(x,u)dx\bigg\}\, du\bigg] ds.
\end{aligned}
\end{equation*}
We first focus on $P_1(w)$.  


\subsection{Expansion of $P_1(w)$}
Using \eqref{h:11a}, we have as $w\rightarrow\infty$,
\begin{equation*}
	P_1(w) = \frac{\pi^{3/2}}{2^{10/3}w^{3/2}}e^{-\frac{w^3}{24}}\int_{-\infty}^\infty \left[\fcal_1(s) \sum_{n=0}^{\infty}\frac{1}{w^n}\int_s^{\infty}e^{-\frac{wu}{4}-\int_u^{\infty}q(x)dx}G_n(u,w)du\right]\,ds\,,
\end{equation*}
where
\begin{equation}\label{Fn}
	G_n(s,w) = Q_n\big(s,q(s),q'(s)\big)\Big(\Ai'\big(2^{-2/3}s+2^{-8/3}w^2\big)-\frac{w}{2^{4/3}}\Ai\big(2^{-2/3}s+2^{-8/3}w^2\big)\Big).
\end{equation}
Let us write this as
\begin{equation}\label{hcal}
P_1(w)=\frac{\pi^{3/2}}{2^{10/3}w^{3/2}}e^{-\frac{w^3}{24}} \sum_{n=0}^\infty \frac{1}{w^n} \iint_{\Omega}\hcal_n(s,u)\,du\,ds\,,
\end{equation}
with
\begin{equation*}
	\hcal_n(s,u) =\fcal_1(s)e^{-\frac{wu}{4}-\int_u^{\infty}q(x)\,dx}G_n(u,w);\ \ \Omega=\big\{(s,u)\in\mathbb{R}^2:\ -\infty<s\leq u<\infty\big\}.
\end{equation*}
At this point recall the asymptotics from \cite{BBD}, as $s\rightarrow -\infty$,
\begin{eqnarray}\label{TW1as}
	\fcal_1(s) &=& \frac{\tau_1 e^{-\frac{|s|^3}{24}-\frac{|s|^{3/2}}{3\sqrt{2}}}}{|s|^{1/16}}\Big(1-\frac{|s|^{-3/2}}{24\sqrt{2}}+\frac{55|s|^{-3}}{2304}-\frac{10675|s|^{-9/2}}{165888\sqrt{2}}\nonumber\\
	&&\qquad\qquad+\frac{3970225|s|^{-6}}{31850496}+O\big(|s|^{-15/2}\big)\Big),\label{F1}\hspace{0.5cm}\tau_1=\frac{e^{\frac{1}{2}\z'(-1)}}{2^{11/48}},
\end{eqnarray} 
as well as for $u\rightarrow-\infty$,
\begin{equation*}
	e^{-\int_u^{\infty}q(x)dx} = \frac{1}{\sqrt{2}}e^{-\frac{\sqrt{2}}{3}|u|^{3/2}}\Big(1-\frac{|u|^{-3/2}}{12\sqrt{2}}+\frac{|u|^{-3}}{576}-\frac{2629|u|^{-9/2}}{20736\sqrt{2}}
	+\frac{10513|u|^{-6}}{1990656}+O\big(|u|^{-15/2}\big)\Big).
\end{equation*}
Since $G_n(s,w)$ has at most polynomial growth as $s\rightarrow-\infty$ and is exponentially decaying as $s\rightarrow\infty$, we see that $\hcal_n(s,u)$ is indeed integrable over $\Omega$ (here we used that $\fcal_1(s)e^{-\int_u^{\infty}q(x)dx}$ is bounded as $s,u\rightarrow+\infty$). Thus we can apply Fubini's theorem in \eqref{hcal} and obtain
\begin{equation}\label{h:1a}
	\iint_{\Omega}\hcal_n(s,u)du\,ds = \int_{-\infty}^{\infty}F(u)e^{-\frac{wu}{4}-\int_u^{\infty}q(x)dx}G_n(u,w)du,
\end{equation}
where 
\begin{equation*}\label{h:1b}
F(u) := \int_{-\infty}^u\fcal_1(s)\,ds.
\end{equation*}
For large negative $u$, we can integrate \eqref{F1} and obtain the following proposition.
\begin{prop}\label{prop1}
As $u\rightarrow -\infty$, 
\begin{eqnarray}\label{Fas}
F(u)&=&\frac{8\tau_1}{|u|^{33/16}}e^{-\frac{|u|^3}{24}-\frac{|u|^{3/2}}{3\sqrt{2}}}\bigg(1-\frac{97\sqrt{2}}{48|u|^{3/2}}-\frac{19337}{2304|u|^3}+\frac{24666605\sqrt{2}}{331776|u|^{9/2}}\nonumber\\
&&\hspace{6cm}+\frac{1358238769}{31850496|u|^6}+O\big(|u|^{-15/2}\big)\bigg),\label{exp:9}
\end{eqnarray}
and $\tau_1$ is given in \eqref{F1}.
\end{prop}
\begin{proof}
We start with \eqref{F1} and obtain for $|u|$ sufficiently large after a change of variables
\begin{eqnarray}
	F(u) &=& \tau_1\int_{|u|}^{\infty}\frac{e^{-\frac{x^3}{24}-\frac{x^{3/2}}{3\sqrt{2}}}}{x^{1/16}}\Big(1-\frac{x^{-3/2}}{24\sqrt{2}}+\frac{55x^{-3}}{2304}-\frac{10675x^{-9/2}}{165888\sqrt{2}}\nonumber\\
	&&\hspace{6cm}+\frac{3970225x^{-6}}{31850496}+O\big(x^{-15/2}\big)\Big)dx.\label{p:1}
\end{eqnarray}
Put $t=\frac{x^3}{24}+\frac{x^{3/2}}{3\sqrt{2}},$ or equivalently,
\begin{equation}\label{subs}
	x = 8^{1/3}\big(-1+\sqrt{1+3t}\big)^{2/3},
\end{equation}
so that as $t\rightarrow\infty$,
\begin{equation*}	
	x=24^{1/3}t^{1/3}\bigg(1-\frac{2t^{-1/2}}{3\sqrt{3}}+\frac{2}{27t}+\frac{5t^{-3/2}}{243\sqrt{3}}-\frac{16}{2187t^2}+O\big(t^{-5/2}\big)\bigg),
\end{equation*}
and
\begin{equation*}
	dx = \frac{24^{1/3}}{3t^{2/3}}\bigg(1+\frac{t^{-1/2}}{3\sqrt{3}}-\frac{4}{27t}-\frac{35t^{-3/2}}{486\sqrt{3}}+\frac{80}{2187t^2}+O\big(t^{-5/2}\big)\bigg)dt.
\end{equation*}
Substituting \eqref{subs} into \eqref{p:1}, we obtain
\begin{eqnarray*}
	F(u) &=& \frac{\tau_1 24^{5/16}}{3}\int_{\frac{|u|^3}{24}+\frac{|u|^{3/2}}{3\sqrt{2}}}^{\infty}e^{-t}t^{-11/16}\Big(1+\frac{35t^{-1/2}}{96\sqrt{3}}-\frac{8129}{55296t}-\frac{1278127t^{-3/2}}{15925248\sqrt{3}}\\
	&&\hspace{7cm}+\frac{661165345}{18345885696t^2}+O\big(t^{-5/2}\big)\Big)\,dt\,,
\end{eqnarray*}
which can be readily integrated by parts to produce \eqref{exp:9}.
\end{proof}
Notice in particular that by \eqref{TW1as} and \eqref{Fas}, as $u\rightarrow-\infty$,
\begin{equation*}
	\frac{\fcal_1(u)}{F(u)}=\frac{u^2}{8}\bigg(1+\frac{2\sqrt{2}}{|u|^{3/2}}+\frac{33}{2|u|^3}-\frac{97\sqrt{2}}{4|u|^{9/2}}
	-\frac{4791}{16|u|^6}+O\big(|u|^{-15/2}\big)\bigg).
\end{equation*}

We split the integral in \eqref{h:1a} in the following way. For some $\alpha>0$, which will be specified more precisely in the following, let
\begin{equation}\label{Hsplit}
	\iint_{\Omega}\hcal_n(s,u)\,du\, ds = I_{1,n}+I_{2,n},
	\end{equation}
where
	\begin{equation*}
	\begin{aligned}
	I_{1,n}&=\int_{-\infty}^{-w^{\alpha}}F(u)e^{-\frac{wu}{4}-\int_u^{\infty}q(x)dx}G_n(u,w)du\,, \\
	I_{2,n}&=\int_{-w^{\alpha}}^{\infty}F(u)e^{-\frac{wu}{4}-\int_u^{\infty}q(x)dx}G_n(u,w)du\,.
	\end{aligned}
\end{equation*}
For $I_{1,n}$, choose $w>0$ sufficiently large and replace $F(u)$ by its asymptotics \eqref{exp:9}:
\begin{eqnarray}\label{I1:1}
	I_{1,n} &=& 8\tau_1\int_{-\infty}^{-w^{\alpha}}e^{\frac{u^3}{24}-\frac{(-u)^{3/2}}{3\sqrt{2}}-\frac{wu}{4}-\int_u^{\infty}q(x)dx}(-u)^{-33/16}G_n(u,w)\Big(1+O\big((-u)^{-3/2}\big)\Big)du\nonumber\\
	&=& 8\tau_1\int_{w^{\alpha}}^{\infty}e^{-\left(\frac{t^3}{24}+\frac{t^{3/2}}{3\sqrt{2}}-\frac{wt}{4}+\int_{-t}^{\infty}q(x)dx\right)}t^{-33/16}G_n(-t,w)\Big(1+O\big(t^{-3/2}\big)\Big)dt.
\end{eqnarray}
Since (see \cite{BBD})
\begin{equation*}
	\int_{-t}^{\infty}q(x)dx = \frac{\sqrt{2}}{3}t^{3/2}+\frac{1}{2}\log2 - \int_{-\infty}^{-t}\bigg(q(x)-\sqrt{\frac{-x}{2}}\bigg)dx\,,
\end{equation*}
we can use \eqref{PIIasy} and obtain
\begin{equation*}
	\int_{-t}^{\infty}q(x)dx = \frac{\sqrt{2}}{3}t^{3/2}+\frac{1}{2}\log2+\frac{\sqrt{2}}{24}t^{-3/2}+O\big(t^{-9/2}\big),\ \ t\rightarrow\infty.
\end{equation*}
To evaluate \eqref{I1:1} asymptotically, we use these asymptotics to obtain
\begin{equation*}
	I_{1,n} = 4\sqrt{2}\tau_1\int_{w^{\alpha}}^{\infty}e^{-\big(\frac{t^3}{24}+\frac{t^{3/2}}{\sqrt{2}}-\frac{wt}{4}\big)}t^{-33/16}G_n(-t,w)\left(1+O\big(t^{-3/2}\big)\right)dt\,.
\end{equation*}
In order to ensure convergence of the integral, we need to impose 
\begin{equation*}
	1+\frac{24}{\sqrt{2}}t^{-3/2}-\frac{6w}{t^2}>0,\ \ \textnormal{as}\ \ w\rightarrow\infty,
\end{equation*}
which can be guaranteed for any $\alpha>\frac{1}{2}$. Let us use the change of variables
\begin{equation*}
	u = \frac{t^3}{24}+\frac{t^{3/2}}{\sqrt{2}}-\frac{wt}{4},\qquad t>0.
\end{equation*}
As $u\rightarrow\infty$
\begin{equation}\label{chtu}
	t = 24^{1/3}u^{1/3}\bigg(1-\frac{\sqrt{12}}{3}u^{-1/2}+\frac{24^{1/3}w}{12}u^{-2/3}+O\big(u^{-1}\big)\bigg),
\end{equation}
and we notice that the only critical point of the function $u=u(t)$ as $w\rightarrow\infty$ is given by
\begin{equation*}
	t_0= \sqrt{2w}\Big(1+O\big(w^{-3/4}\big)\Big),
\end{equation*}
which by the choice $\al>\frac{1}{2}$ lies not in the domain of integration. Moreover,
\begin{equation*}
	dt = \frac{24^{1/3}}{3}u^{-2/3}\bigg(1-\frac{\sqrt{12}}{6}u^{-1/2}-\frac{24^{1/3}w}{12}u^{-2/3}+O\big(u^{-1}\big)\bigg)du\,,
\end{equation*}
and we conclude
\begin{equation}\label{I1n}
	I_{1,n} = \frac{4}{3}\sqrt{2}\tau_124^{-17/48}\int_{\frac{w^{3\alpha}}{24}+\frac{w^{3\alpha/2}}{\sqrt{2}}-\frac{w^{\alpha+1}}{4}}^{\infty}e^{-u}u^{-65/48}G_n\big(-t(u),w\big)\Big(1+O\big(u^{-1/2}\big)\Big)du.
\end{equation}
Since $G_n(-t(u),w)$ in the integral above grows at most like a polynomial as $w\rightarrow\infty$, we have the  estimate
\begin{equation}\label{In1esti}
	|I_{1,n}| \leq ce^{-\frac{w^{3\alpha}}{24}-\frac{w^{3\alpha/2}}{\sqrt{2}}+\frac{w^{\alpha+1}}{4}}|R_n\big(w^{\alpha}\big)|, \quad \textrm{as} \ w\rightarrow\infty,
\end{equation}
for some polynomial $R_n$ and $c\in \R$.

\smallskip

For $I_{2,n}$,
\begin{equation*}
	I_{2,n} = \int_{-w^{\alpha}}^{\infty}F(u)e^{-\frac{wu}{4}-\int_u^{\infty}q(x)dx}G_n(u,w)du,
\end{equation*}
we choose more specific bounds on $\alpha$, namely $1<\alpha<2$. Then, as $w\rightarrow\infty$,
\begin{equation*}
	2^{-2/3}u +2^{-8/3}w^2 \rightarrow+\infty,
\end{equation*}
and we can use the expansion \eqref{pr:15} for $a(u, -w)$.  This gives
\begin{equation*}
	G_n(u,w) = -2^{-2/3}\sqrt{\frac{w}{\pi}}Q_n\big(u,q(u),q'(u)\big)\Pi(u,w)e^{-\frac{2}{3}(2^{-2/3}u+2^{-8/3}w^2)^{3/2}},\ \ w\rightarrow\infty,
\end{equation*}
where $\Pi(u,w)$ can be read from \eqref{pr:15}:
\begin{eqnarray}
	\Pi(u,w) &=& \frac{1}{2}\sum_{n,m=0}^{\infty}\bigg(\frac{4u}{w^2}\bigg)^n\bigg(-\frac{3}{2}\bigg)^m\Big(2^{-2/3}u+2^{-8/3}w^2\Big)^{-\frac{3m}{2}}\bigg[\binom{1/4}{n}d_m+\binom{-1/4}{n}c_m\bigg]\nonumber\\
	&=&1+\frac{1}{3w^3}+\frac{u^2}{2w^4}+O\big(w^{-6}\big),\hspace{0.5cm} w\rightarrow+\infty.\label{Pi}
\end{eqnarray}
In general, the coefficient of $w^{-k}$ in the latter series is a polynomial in $u$ of degree at most $k/2$. Let us now write
\begin{equation*}
	I_{2,n}=-2^{-2/3}\sqrt{\frac{w}{\pi}}\,L_n(w),
\end{equation*}
where
\begin{equation*}
	L_n(w)=\int_{-w^{\alpha}}^{\infty}e^{-wH(u)}Q_n\big(u,q(u),q'(u)\big)\Pi(u,w)\,du,
\end{equation*}
and 
\begin{equation*}
	H(u)\equiv H(u|w) = -\frac{1}{w}\ln F(u)+\frac{u}{4}+\frac{1}{w}\int_u^{\infty}q(x)dx+\frac{w^2}{24}\bigg(1+\frac{4u}{w^2}\bigg)^{3/2},\ u\in\mathbb{R}.
\end{equation*}
The derivative of the function $H(u)$ is
\begin{equation*}
	H'(u) = -\frac{\fcal_1(u)}{wF(u)}+\frac{1}{4}-\frac{q(u)}{w}+\frac{1}{4}\sqrt{1+\frac{4u}{w^2}}.
\end{equation*}
As $w\rightarrow\infty$, the only zero $u_0$ of the function $H'(u)$ will lie in a neighborhood of $u=-\infty$, so we can use Proposition \ref{prop1}, \eqref{F1} and \eqref{PIIasy} to determine it. The solution is 
\begin{equation}\label{cp}
	u_0 = -2\sqrt{w}\Big(1-\frac{3}{2w^{3/4}}-\frac{65}{32w^{3/2}}-\frac{3}{8w^{9/4}}+O\big(w^{-3}\big)\Big),\ \ w\rightarrow\infty\,,
\end{equation}
and the integral $L_n(w)$ can be evaluated as $w\rightarrow\infty$ by Laplace's method. For the convenience of the reader, we have included a short review of Laplace's method in Appendix \ref{app2}. Notice that as $w\rightarrow\infty$,
\begin{eqnarray*}
	e^{-\frac{wu_0}{4}-\frac{w^3}{24}\big(1+\frac{4u_0}{w^2}\big)^{3/2}} &=& e^{-\frac{w^3}{24}+w^{3/2}-\frac{3}{2}w^{3/4}-\frac{97}{32}}\bigg(1+\frac{21}{8w^{3/4}}+O\big(w^{-3/2}\big)\bigg),\\
	e^{-\int_{u_0}^{\infty}q(x)dx}&=&\frac{e^3}{\sqrt{2}}e^{-\frac{4}{3}w^{3/4}}\bigg(1+\frac{35}{12w^{3/4}}+O\big(w^{-3/2}\big)\bigg),\\
	F(u_0) &=&\frac{2^{15/16}\tau_1}{w^{33/32}}e^{-\frac{w^{3/2}}{3}+\frac{5}{6}w^{3/4}+\frac{41}{32}}\bigg(1-\frac{25}{24w^{3/4}}+O\big(w^{-3/2}\big)\bigg),
\end{eqnarray*}
which implies 
\begin{equation}\label{L1}
	e^{-wH(u_0|w)} = \frac{2^{5/24}e^{\frac{1}{2}\zeta'(-1)}e^{5/4}}{w^{33/32}}e^{-\frac{w^3}{24}+\frac{2}{3}w^{3/2}-2w^{3/4}}\bigg(1+\frac{9}{2}w^{-3/4}+O\big(w^{-3/2}\big)\bigg).
\end{equation}
Also,
\begin{equation}\label{Hdprime}
	\big(wH''(u_0)\big)^{-1/2}=\frac{\sqrt{2}}{w^{1/4}}\bigg(1+\frac{3}{8w^{3/4}}+O\big(w^{-3/2}\big)\bigg),\quad \textrm{as} \ w\rightarrow\infty.
\end{equation}
Combining \eqref{L1} and \eqref{Hdprime},  Laplace's method gives
\begin{eqnarray*}
	L_n(w) &=& \frac{2^{29/54}\sqrt{\pi}e^{\frac{1}{2}\zeta'(-1)}e^{5/4}}{w^{41/32}} e^{-\frac{w^3}{24}+\frac{2}{3}w^{3/2}-2w^{3/4}}Q_n\big(u_0,q(u_0),q'(u_0)\big)\\
	&&\times\Pi(u_0,w)\bigg(1+\frac{39}{8w^{3/4}}+O\big(w^{-3/2}\big)\bigg).
\end{eqnarray*}
It follows that
\begin{eqnarray*}
	I_{2,n}&=&-\frac{2^{13/24}e^{\frac{1}{2}\zeta'(-1)}e^{5/4}}{w^{25/32}} e^{-\frac{w^3}{24}+\frac{2}{3}w^{3/2}-2w^{3/4}}Q_n\big(u_0,q(u_0),q'(u_0)\big)\\
	&&\times\Pi(u_0,w)\bigg(1+\frac{39}{8w^{3/4}}+O\big(w^{-3/2}\big)\bigg),\hspace{0.5cm}w\rightarrow\infty.
\end{eqnarray*}
From \eqref{Pi} and the property that the coefficient of $w^{-k}$ in the asymptotic series for $\Pi(u,w)$ is a polynomial in $u$ of degree at most $k/2$, we find
\begin{equation}\label{pesti}
	\Pi(u_0,w) = 1+O\big(w^{-3}\big),\ w\rightarrow\infty.
\end{equation}
Plugging into  \eqref{hcal} and \eqref{Hsplit}, we obtain   
\begin{eqnarray}
	P_1(w) &=& \frac{\pi^{3/2}}{2^{10/3}w^{3/2}}e^{-\frac{w^3}{24}}\sum_{n=0}^{\infty}\frac{1}{w^n}\Big(I_{1,n}+I_{2,n}\Big)\nonumber\\
	&=& - \frac{\pi^{3/2}e^{\frac{1}{2}\zeta'(-1)}e^{5/4}}{w^{73/32}2^{67/24}}e^{-\frac{w^3}{12}+\frac{2}{3}w^{3/2}-2w^{3/4}}\bigg(1+\frac{39}{8w^{3/4}}+O\big(w^{-3/2}\big)\bigg) \nonumber\\
	&& \qquad \times  \sum_{n=0}^{\infty}Q_n\big(u_0,q(u_0),q'(u_0)\big)w^{-n},
	\label{L2}
\end{eqnarray}
where we have used the estimate \eqref{In1esti} which shows that $I_{1,n}$ does not contribute to the leading order asymptotics when $\al>1$. Next by Proposition \ref{prop0}, as $w\rightarrow\infty$,
\begin{equation*}
	\sum_{n=0}^{\infty}Q_n\big(u_0,q(u_0),q'(u_0)\big)w^{-n} = Q_0\big(u_0,q(u_0),q'(u_0)\big)+\frac{1}{w}Q_1\big(u_0,q(u_0),q'(u_0)\big)+O\big(w^{-7/4}\big),
\end{equation*}
and from \eqref{Q01},
\begin{equation*}
	Q_0\big(u_0,q(u_0),q'(u_0)\big) = -v(u_0) = -\frac{1}{2w^{1/4}}\bigg(1+\frac{5}{8w^{3/4}}+O\big(w^{-3/2}\big)\bigg),\hspace{0.5cm}\textrm{as} \ w\rightarrow\infty,
\end{equation*}
and
\begin{equation*}
	Q_1\big(u_0,q(u_0),q'(u_0)\big) = -\frac{3}{2}+O\big(w^{-3/4}\big).
\end{equation*}
Thus
\begin{equation}\label{PestiII}
		\sum_{n=0}^{\infty}Q_n\big(u_0,q(u_0),q'(u_0)\big)w^{-n} =-\frac{1}{2w^{1/4}}\bigg(1+\frac{29}{8w^{3/4}}+O\big(w^{-3/2}\big)\bigg),\hspace{0.5cm}w\rightarrow\infty.
\end{equation}
Combining this with \eqref{L2} we find, as $w\rightarrow\infty$,
\begin{equation}\label{P1fin}
	P_1(w) = \kappa e^{-\frac{w^3}{12}+\frac{2}{3}w^{3/2}-2w^{3/4}}w^{-81/32}\bigg(1+\frac{17}{2w^{3/4}}+O\big(w^{-3/2}\big)\bigg),
\end{equation}
with
\begin{equation}\label{P1fincon}
	\kappa=2^{-91/24}\pi^{3/2}e^{\frac{1}{2}\zeta'(-1)}e^{5/4}.
\end{equation}
We now turn our attention to $P_2$.

 
\subsection{Expansion of $P_2(w)$}
Following the philosophy of the previous subsection, write
\begin{equation*}
	P_2(w) = \frac{\pi^{3/2}}{2^{10/3}w^{3/2}}e^{-\frac{w^3}{24}}\sum_{n=0}^{\infty}\frac{1}{w^n}\iiint_{\widehat{\Omega}}\widehat{\hcal}_n(s,u,x)dx\,du\,ds\,,
\end{equation*}
with
\begin{equation*}
	\widehat{\hcal}_n(s,u,x) = \fcal_1(s)e^{-\frac{w(u+2x)}{4}-\int_u^{\infty}q(t)dt}\widetilde{G}_n(u,w,x)\,,
\end{equation*}
\begin{equation*}
	\widetilde{G}_n(u,w,x) = Q_n\big(u,q(u),q'(u)\big)\left(\Ai'\bigg(\frac{u+2x}{2^{2/3}}+\frac{w^2}{2^{8/3}}\bigg)-\frac{w^2}{2^{4/3}}\Ai\bigg(\frac{u+2x}{2^{2/3}}+\frac{w^2}{2^{8/3}}\bigg)\right)K(x,u)\,,
\end{equation*}
and the domain of integration
\begin{equation*}
	\widehat{\Omega} = \big\{(s,u,x)\in\mathbb{R}^3:\ -\infty<s\leq u<\infty,\ x\geq 0\big\}.
\end{equation*}
Recalling the behavior of $K(x,u)$ as we approach the boundary of $\widehat{\Omega}$, we can again apply Fubini's theorem to obtain
\begin{equation*}\label{P2:1}
	\iiint_{\widehat{\Omega}}\widehat{\hcal}_n(s,u,x)dx\,du\,ds = \int_{-\infty}^{\infty}F(u)e^{-\frac{wu}{4}-\int_u^{\infty}q(t)dt}\widehat{G}_n(u,w)du\,,
\end{equation*}
where (compare \eqref{Fn})
\begin{equation}\label{P2:2}
\begin{aligned}
	\widehat{G}_n(u,w) =& \int_0^{\infty}\widetilde{G}_n(u,w,x)K(x,u)dx \\
	 =&Q_n\big(u,q(u),q'(u)\big) \\
	& \times\int_0^{\infty}e^{-\frac{wx}{2}}\left[\Ai'\left(\frac{u+2x}{2^{2/3}}+\frac{w^2}{2^{8/3}}\right)-\frac{w^2}{2^{4/3}}\Ai\left(\frac{u+2x}{2^{2/3}}
	+\frac{w^2}{2^{8/3}}\right)\right]K(x,u)dx.
	\end{aligned}
\end{equation}
Introducing $1<\al<2$, we again split the integral as
\begin{eqnarray*}
	&&\iiint_{\widehat{\Omega}}\widehat{\hcal}_n(s,u,x)dx\,du\,ds = \hat{I}_{1,n}+\hat{I}_{2,n}\\
	&\equiv& \int_{-\infty}^{-w^{\al}}F(u)e^{-\frac{wu}{4}-\int_u^{\infty}q(t)dt}\widehat{G}_n(u,w)du+\int_{-w^{\al}}^{\infty}F(u)e^{-\frac{wu}{4}-\int_u^{\infty}q(t)dt}\widehat{G}_n(u,w)du.
\end{eqnarray*}
For $\hat{I}_{1,n}$ we follow the same steps as for $I_{1,n}$ and conclude (see \eqref{I1n})
\begin{equation*}
	\hat{I}_{1,n} = \frac{4}{3}\sqrt{2}\tau_124^{-17/48}\int_{\frac{w^{3\al}}{24}+\frac{w^{3\al/2}}{\sqrt{2}}-\frac{w^{\al+1}}{4}}^{\infty}e^{-u}u^{-65/48}\widehat{G}_n\big(-t(u),w)\big)\Big(1+O\big(u^{-1/2}\big)\Big)du\,,
\end{equation*}
with $t=t(u)$ given in \eqref{chtu}. Also here, $\widehat{G}_n(u,w)$ grows at most like a polynomial as $u\rightarrow-\infty$, hence we have the estimate
\begin{equation*}\label{P2:3}
	|\hat{I}_{1,n}|\leq \hat{c}e^{-\frac{w^{3\al}}{24}-\frac{w^{3\al/2}}{\sqrt{2}}+\frac{w^{\al+1}}{4}}|\widehat{R}_n(w^{\al})|,\hspace{0.5cm} \textrm{as} \ w\rightarrow\infty,
\end{equation*}
for some polynomial $\widehat{R}_n$. 

For $\hat{I}_{2,n}$ we can replace the Airy functions in \eqref{P2:2} with their large argument asymptotics, giving
\begin{equation*}
	\widehat{G}_n(u,w) = -2^{-2/3}\sqrt{\frac{w}{\pi}}\int_0^{\infty}e^{-wJ(u,x|w)}\Pi(u+2x,w)K(x,u)dx\,,
\end{equation*}
with
\begin{equation*}
	J(u,x)\equiv J(u,x|w) = \frac{x}{2}+\frac{w^2}{24}\bigg(1+\frac{4(u+2x)}{w^2}\bigg)^{3/2}.
\end{equation*}
Plugging this into the integral $\hat{I}_{2,n}$ gives
\begin{equation*}
	\hat{I}_{2,n} = -2^{-2/3}\sqrt{\frac{w}{\pi}}\iint_De^{-w(\hat{H}(u|w)+J(u,x|w))}Q_n\big(u,q(u),q'(u)\big)\Pi(u+2x,w)K(x,u)dx\,du\,,
\end{equation*}
where the domain of integration $D$ equals
\begin{equation*}
	D=\big\{(u,x)\in\mathbb{R}^2:\ -w^{\al}<u<\infty,\ x\geq 0\big\},
\end{equation*}
where $\al \in(1,2)$ is fixed, 
and $\hat{H}=\hat{H}(u|w)$ is given by
\begin{equation*}
	\hat{H}(u|w) = -\frac{1}{w}\ln F(u)+\frac{u}{4}+\frac{1}{2}\int_u^{\infty}q(x)\,dx.
\end{equation*}
Since
\begin{eqnarray*}
	\frac{\partial}{\partial u}\big(\hat{H}(u|w)+J(u,x|w)\big) &=& -\frac{\fcal_1(u)}{wF(u)}+\frac{1}{4}-\frac{q(u)}{w}+\frac{1}{4}\sqrt{1+\frac{4(u+2x)}{w^2}}\,,\\
	\frac{\partial}{\partial x}\big(\hat{H}(u|w)+J(u,x|w)\big) &=& \frac{1}{2}+\frac{1}{2}\sqrt{1+\frac{4(u+2x)}{w^2}}>0\,,
\end{eqnarray*}
we see that $\hat{H}(u|w)+J(u,x|w)$ has no critical point in the interior of $D$. However its partial derivative with respect to $u$ vanishes in a neigbhorhood of $(-\infty,x)$. More precisely, it vanishes for $u=u_0(x)$, where
\begin{equation*}\label{u0x}
	u_0(x) = -2\sqrt{w}\bigg(1-\frac{3}{2w^{3/4}}-\frac{65}{32w^{3/2}}+\frac{x}{w^2}-\frac{3}{8w^{9/4}}+\frac{3x}{4w^{11/4}}+O\big(w^{-3}\big)\bigg),\hspace{0.2cm}\textrm{as} \ w\rightarrow\infty.
\end{equation*}
Notice that $u_0(0)=u_0$ as given in \eqref{cp}. Laplace's method now indicates that we need to expand the integral $\hat{I}_{2,n}$ in a neighborhood of the point $(u_0(0),0)$. First,
\begin{eqnarray*}
	\hat{H}(u|w)+J(u,x|w) &=& H(u_0|w)+J_x(u_0,0|w)x\\
	&&+\frac{1}{2}(u-u_0,x)M(u_0,0|w)\binom{u-u_0}{x}+O\big((u-u_0)^3x^3\big),
\end{eqnarray*}
which is valid in a neighborhood of $(u_0(0),0)$, with $M(u,x|w)$ denoting the Hessian matrix of $\hat{H}(u|w)+J(u,x|w)$. Since the mixed partial $J_{xu}$ differs from $J_{xx}$ by a factor of $1/2$, we have
\begin{equation*}
	M(u_0,0|w) = \begin{pmatrix}
	H''(u_0) & \frac{1}{2}J_{xx}(u_0,0)\\
	\frac{1}{2}J_{xx}(u_0,0) & J_{xx}(u_0,0)\\
	\end{pmatrix},\hspace{0.5cm} J_{xx}(u_0,0) = \frac{2}{w^2}\bigg(1+\frac{4u_0}{w^2}\bigg)^{-1/2},
\end{equation*}
and therefore,
\begin{eqnarray*}
	M(u_0,0|w)&=& \frac{1}{2\sqrt{w}}\bigg[\begin{pmatrix}
	1 & 0\\
	0 & 0\\
	\end{pmatrix}-\frac{3}{4w^{3/4}}\begin{pmatrix}
	1 & 0\\
	0 & 0\\
	\end{pmatrix} +\frac{1}{2w^{3/2}}\begin{pmatrix}
	-7 & 4\\
	4 & 8\\
	\end{pmatrix}-\frac{9}{64w^{9/4}}\begin{pmatrix}
	1 & 0\\
	0 & 0\\
	\end{pmatrix}\\
	&&+O\big(w^{-5/2}\big)\bigg],\hspace{0.5cm}\textrm{as } \ w\rightarrow\infty.
\end{eqnarray*}
Here we have used that (compare \eqref{Hdprime})
\begin{equation*}
	H''(u_0) = \frac{1}{2\sqrt{w}}\bigg(1-\frac{3}{4w^{3/4}}-\frac{7}{2w^{3/2}}-\frac{9}{64w^{9/4}}+O\big(w^{-5/2}\big)\bigg),\hspace{0.5cm}\textrm{as } \ w\rightarrow\infty.
\end{equation*}
Moving on with Laplace's method, we have
\begin{eqnarray*}
	\hat{I}_{2,n}&=&-2^{-2/3}\sqrt{\frac{w}{\pi}}\iint_De^{-w(\hat{H}(u,w)+J(u,x|w))}Q_n\big(u,q(u),q'(u)\big)\Pi(u+2x,w)K(x,u)dxdx\\
	&=&-2^{-2/3}\sqrt{\frac{w}{\pi}}e^{-wH(u_0|w)}Q_n\big(u_0,q(u_0),q'(u_0)\big)\Pi(u_0,w)K(0,u_0)\int_0^{\infty}e^{-wJ_x(u_0,0|w)x}dx\\
	&&\times\int_{-\infty}^{\infty}e^{-\frac{\sqrt{w}}{4}\left(1-\frac{3}{4w^{3/4}}\right)u^2}du\Big(1+O\big(w^{-1}\big)\Big).
\end{eqnarray*}
The integrals in the above expression can be expanded for large $w$ as
\begin{equation*}
	\int_0^{\infty}e^{-wJ_x(u_0,0|w)x}dx = \frac{1}{w}\bigg(1-\frac{2}{w^{3/2}}+O\big(w^{-9/4}\big)\bigg),
\end{equation*}
and
\begin{equation*}
	\int_{-\infty}^{\infty}e^{-\frac{\sqrt{w}}{4}\left(1-\frac{3}{4w^{3/4}}\right)u^2}du= \frac{2\sqrt{\pi}}{w^{1/4}}\bigg(1+\frac{3}{8w^{3/4}}+\frac{27}{128w^{3/2}}+O\big(w^{-9/4}\big)\bigg).
\end{equation*}
Combining with \eqref{L1}, we obtain
\begin{eqnarray*}	\hat{I}_{2,n}&=&-\frac{2^{13/24}e^{\frac{1}{2}\zeta'(-1)}e^{5/4}}{w^{57/32}}e^{-\frac{w^3}{24}+\frac{2}{3}w^{3/2}-2w^{3/4}}Q_n\big(u_0,q(u_0),q'(u_0)\big)\\
	&&\hspace{3cm} \times\Pi(u_0,w)K(0,u_0)\bigg(1+\frac{39}{8w^{3/4}}+O\big(w^{-3/2}\big)\bigg).
\end{eqnarray*}
Applying the same arguments as in the computation of $P_1(w)$, we get from \eqref{pesti} and \eqref{PestiII}
\begin{eqnarray}
	P_2(w) &=&\frac{\pi^{3/2}}{2^{10/3}w^{3/2}}e^{-\frac{w^3}{24}}\sum_{n=0}^{\infty}\frac{1}{w^n}\bigg(\hat{I}_{1,n}+\hat{I}_{2,n}\bigg)\nonumber\\
	 &=&\kappa e^{-\frac{w^3}{12}+\frac{2}{3}w^{3/2}-2w^{3/4}}w^{-113/32}K(0,u_0)\bigg(1+\frac{17}{2w^{3/4}}+O\big(w^{-3/2}\big)\bigg),\label{P2:4}
\end{eqnarray}
with $\kappa$ given in \eqref{P1fincon}. At this point recall \eqref{pr:22}, which gives
\begin{equation*}
	K(0,u_0) = w\bigg(1-\frac{2}{w^{3/4}}+O\big(w^{-3/2}\big)\bigg),\hspace{0.5cm} \textrm{as} \ w\rightarrow\infty.
\end{equation*}
Inserting this asymptotics into \eqref{P2:4}, we find
\begin{equation}\label{P2:5}
	P_2(w) = \kappa e^{-\frac{w^3}{12}+\frac{2}{3}w^{3/2}-2w^{3/4}}w^{-81/32}\bigg(1+\frac{13}{2w^{3/4}}+O\big(w^{-3/2}\big)\bigg).
\end{equation}


\subsection{Asymptotic expansion of $P(w)$}
Since $P(w)=P_1(w)+P_2(w)$, we add \eqref{P1fin} and \eqref{P2:5}, which leads us to
\begin{equation*}\label{Pexp}
	P(w) = 2\kappa e^{-\frac{w^3}{12}+\frac{2}{3}w^{3/2}-2w^{3/4}}w^{-81/32}\bigg(1+\frac{15}{2w^{3/4}}+O\big(w^{-3/2}\big)\bigg),\hspace{0.5cm}w\rightarrow\infty\,,
\end{equation*}
with the constant $\kappa$
\begin{equation*}
	\kappa=2^{-91/24}\pi^{3/2}e^{\frac{1}{2}\zeta'(-1)}e^{5/4}.
\end{equation*}
Plugging into \eqref{in:14f}, we get Theorem \ref{main}.


\section{Proof of Corollary \ref{main2}}\label{proofcor}
Since $\hat{P}(\cdot,t)$ is an even function, see \eqref{in:70} or \eqref{in:12} and \eqref{in:13}, we use the symmetry of the density function $\hat{P}(t)$. As $t\rightarrow\infty$,
\begin{eqnarray*}
	\mathbb{P}(|\tcal|>t) &=& 2\int_t^{\infty}\hat{P}(s)ds\\
	&=&2\tau\int_t^{\infty}e^{-\frac{4}{3}\varphi(s)}s^{-81/32}\bigg(1+\frac{15}{4s^{3/4}}+O\big(s^{-3/2}\big)\bigg)ds.
\end{eqnarray*}
Let us use the change of variables
\begin{equation*}
	u = \varphi(t) = t^3-2t^{3/2}+3t^{3/4},
\end{equation*}
hence as $u\rightarrow\infty$
\begin{equation*}
	t = \sqrt[3]{u}\bigg(1+\frac{2}{3u^{1/2}}-\frac{1}{u^{3/4}}+O\big(u^{-1}\big)\bigg),\hspace{0.5cm} dt = \frac{1}{3u^{2/3}}\bigg(1-\frac{1}{3u^{1/2}}+\frac{5}{4u^{3/4}}+O\big(u^{-1}\big)\bigg)du
\end{equation*}
and we obtain
\begin{equation*}
	\mathbb{P}(|\tcal|>t) = \frac{2\tau}{3}\int_{\varphi(t)}^{\infty}e^{-\frac{4}{3}u}u^{-145/96}\bigg(1+\frac{15}{4u^{1/4}}+O\big(u^{-1/2}\big)\bigg)du.
\end{equation*}
The latter integral can be readily integrated by parts and we obtain \eqref{m:3}.


\begin{appendix}

\section{Riemann-Hilbert problem for $\Phi$}\label{app1}
As mentioned in Section \ref{plusinf}, the matrix function $\Psi(\z;s)$ described in that section can be defined as the solution to a Riemann-Hilbert problem (RHP). The exact formulation of this RHP is as follows. Define the rays $\Ga_j$ for $j=1,2,3,4$ as
\begin{equation*}\label{rhp-1}
\begin{aligned}
		\Ga_1&=\{\z : \arg \z = \pi/6\}\,, \quad &\Ga_2&=\{\z : \arg \z = 5\pi/6\}\,, \\ 
		\quad \Ga_3&=\{\z : \arg \z = -5\pi/6\}\,, \quad &\Ga_4&=\{\z : \arg \z = -\pi/6\}\,,
\end{aligned}
\end{equation*}
and give them the orientation from left to right, as shown in Figure \ref{fig1}.  The plus- (resp. minus-) side of the contour is then the side which is lies to the left (resp. right) side of the contour when facing in the direction of the orientation, as shown in Figure \ref{fig1}.  The RHP consists in finding the piecewise analytic $2\times 2$ matrix function $\Psi(\z;s) = \big(\psi_{ij}(\z;s)\big)_{i,j=1}^2$ such that
\begin{enumerate}
	\item $\Psi(\z;s)$ is analytic for $\z \in \C \setminus \cup_{j=1}^4 \Ga_j$.
	\item For $\z \in  \cup_{j=1}^4 \Ga_j$, $\Psi(\z;s)$ takes limiting values from the plus- and minus-sides, and these limiting values are related by
\begin{equation*}\label{rhp-2}
	\Psi_+(\z;s)=\Psi_-(\z;s)j_\Psi(\z)\,,
\end{equation*}
where
\begin{equation*}\label{rhp-3}
	j_\Psi(\z)=\left\{
\begin{aligned}
	j_1=\frac{1}{2}\begin{pmatrix} 3 & i \\ i & 1 \end{pmatrix} \qquad \z \in \Ga_1 \cup \Ga_2\,, \\
	j_2=\frac{1}{2}\begin{pmatrix} 1 & i \\ i & 3 \end{pmatrix} \qquad \z \in \Ga_3 \cup \Ga_4\,. \\
\end{aligned}\right.
\end{equation*}
	\item As $\z \rightarrow \infty$ $$\Psi(\z;s) = \big(I+O(\z^{-1})\big)\begin{pmatrix} \cos\left(\frac{4}{3}\z^3+s\z\right) & \sin\left(\frac{4}{3}\z^3+s\z\right) \\ - 		\sin\left(\frac{4}{3}\z^3+s\z\right) & \cos\left(\frac{4}{3}\z^3+s\z\right)\end{pmatrix}. $$
	\item $\Psi(\z,s)$ is bounded for $\z$ close to the origin.
\end{enumerate}

\begin{figure}[tbh]
  \begin{center}
  \psfragscanon
  \psfrag{1}{\footnotesize{$\Gamma_1$}}
  \psfrag{2}{\footnotesize{$\Gamma_2$}}
  \psfrag{3}{\footnotesize{$\Gamma_3$}}
  \psfrag{4}{\footnotesize{$\Gamma_4$}}
  \psfrag{5}{$\bf{+}$}
  \psfrag{6}{$\bf{-}$}
  \psfrag{7}{$\bf{+}$}
  \psfrag{8}{$\bf{-}$}
  \psfrag{9}{$\bf{+}$}
  \psfrag{10}{$\bf{-}$}
  \psfrag{11}{$\bf{+}$}
  \psfrag{12}{$\bf{-}$}
  \includegraphics[width=8cm,height=4cm]{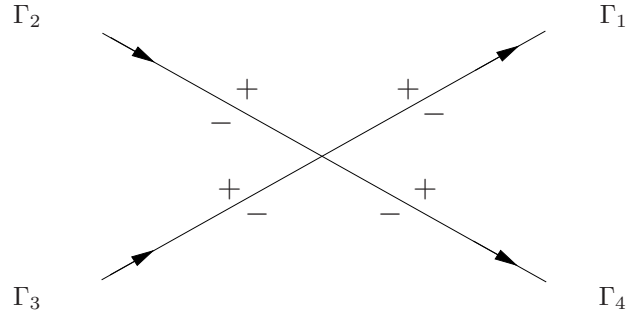}
  \end{center}
  \caption{The jump contour associated with the $\Psi$-RHP}
  \label{fig1}
\end{figure}
This RHP is uniquely solvable \cite{BIK}, and the vector $\Phi(\z,s)$ defined in \eqref{in:6} is the first column of $\Psi$.  That is, 
\begin{equation*}
	\Phi_1(\z;s) = \psi_{11}(\z;s),\hspace{1cm} \Phi_2(\z;s) = \psi_{21}(\z;s).
\end{equation*}
Since the jump matrices satisfy the symmetry relation 
\begin{equation*}\label{rhp-4}
	j_1^{-1} = \sg_3 j_2 \sg_3\,,\hspace{1cm} \sg_3 = \begin{pmatrix}
	1 & 0\\
	0 & -1\\
	\end{pmatrix},
\end{equation*}
we also have
\begin{equation*}
	\Psi(-\z;s)=\sg_3 \Psi(\z;s) \sg_3,\hspace{0.5cm}\z\in\mathbb{C},\ s\in\mathbb{R},
\end{equation*}
which is \eqref{rhp-5}.
\smallskip
	
To justify \eqref{rhp-6}, we appeal to a result from the paper \cite{BR} of Baik and Rains.  In that paper the authors consider a different RHP for a different Lax pair for the Painlev\'{e} II equation.  Our RHP can be transformed to theirs by a series of explicit transformations.  After doing so, equation \eqref{rhp-6} is equivalent to Lemma 2.1 (iii) in \cite{BR}.

\section{A review of Laplace's method}\label{app2}
Let us give a quick review of Laplace's method.  Recall that Laplace's method can be applied to the integral
\begin{equation*}
\int_a^b e^{-wH(u)} f(u)\, du\,,
\end{equation*}
when the function $H(u)$ has a unique minimum at $u_0 \in (a,b)$, where $a$ and/or $b$ may be infinite and $f(u_0)\neq 0$.  If we assume that the functions $H$ and $f$ are analytic, then Laplace's method gives that, as $w \rightarrow\infty$,
\begin{equation*}
\begin{aligned}
&\int_a^b e^{-wH(u)} f(u)\, du= e^{-wH(u_0)} \frac{f(u_0)\sqrt{2\pi}}{\sqrt{wH''(u_0)}}  \Bigg[1+\frac{1}{2wH''(u_0)} \\
&\left. \times\left(\frac{5}{12} \left(\frac{H'''(u_0)}{H''(u_0)}\right)^2 -\frac{H^{(4)}(u_0)}{4H''(u_0)}+\frac{f''(u_0)}{f(u_0)}+\frac{f'''(u_0)}{3f(u_0)}-\frac{H'''(u_0)f'(u_0)}{H''(u_0)f(u_0)}\right)+O(w^{-2})\right].
\end{aligned}
\end{equation*}
Another situation might arise if $H(u)$ has no critical point in the interior of $(a,b)$ but a global minimum is achieved on the boundary, say at $u_0=a$ with $H'(a)>0$ and $f(a)\neq 0$. In this case Laplace's method gives as $w\rightarrow\infty$,
\begin{equation*}
\begin{aligned}
	\int_a^be^{-wH(u)}f(u)du &= \frac{e^{-wH(a)}f(a)}{wH'(a)}\Bigg[1+\frac{1}{w}\bigg(\frac{f'(a)}{f(a)H'(a)}-\frac{H''(a)}{(H'(a))^2}\bigg) \\
	&+\frac{1}{w^2}\bigg(\frac{f''(a)}{f(a)(H'(a))^2}-\frac{3H''(a)f'(a)}{f(a)(H'(a))^3}-\frac{H'''(a)}{(H'(a))^3}\bigg)+O\big(w^{-3}\big)\Bigg].
	\end{aligned}
\end{equation*}
Both of these cases can be generalized to iterated integrals with appropriate assumptions on $H$ and $f$.

\section{The formula of Dotsenko}\label{dot}
Using the so-called replica Bethe ansatz, Dotsenko recently derived a formula for the distribution of the fluctuations of the endpoint of a continuum directed polymer in a random medium.  In this appendix we show that his formula coincides with the known formulas of \cite{M-FQR} and \cite{Sch} for $\tcal$, given in \eqref{in:70} and \eqref{in:12}, respectively.
In his article \cite{Dot}, Dotsenko derives an expression for the distribution function $W(x)$ for the fluctuations of a directed polymer's endpoint.  The conjectured relation between $W(x)$ and the density function for $\tcal$ is
\begin{equation}\label{eq6}
	W(x) = \int_x^{\infty}P(w)dw, \quad \textrm{or equivalently} \quad W(x) = \int_x^{\infty}\int_{-\infty}^{\infty}\frac{1}{4}\hat{P}(2^{-2/3}m,2^{-4/3}t)dm\, dt,
\end{equation}
where $P(w)$ and $\hat{P}(m,t)$ are defined in \eqref{in:14b} and \eqref{in:70}, respectively.   The formula of Dotsenko is (see equations $(6),(7)$ and $(8)$ in \cite{Dot})
\begin{equation}\label{eq1}
	W(x) = \int_{-\infty}^{\infty}\fcal_1(s)\Bigg[\int_0^{\infty}\int_0^{\infty}\rho_s(\omega,\omega')\Phi_{\omega'\omega}(s,x)\,d\omega\, d\omega'\Bigg]ds\,,
\end{equation}
where $\fcal_1$ is the Tracy-Widom GOE distribution function, and $\B_s$ and $\rho_s$ are as defined in \eqref{in:4} and \eqref{in:5}, respectively.  The function $\Phi_{\omega\omega'}$ is given as the integral
\begin{eqnarray}\label{eq3}
	\Phi_{\omega\omega'}(s,x) &=& -\frac{1}{2}\int_0^{\infty}\Bigg[\bigg(\frac{\partial}{\partial\omega}-\frac{\partial}{\partial\omega'}\bigg)\Psi(\omega+\frac{s}{2}+y;x)\Psi(\omega'+\frac{s}{2}+y;-x)\nonumber\\
	&&+\bigg(\frac{\partial}{\partial\omega}+\frac{\partial}{\partial\omega'}\bigg)\Psi(\omega+\frac{s}{2}-y;x)\Psi(\omega'+\frac{s}{2}+y;-x)\Bigg]dy\,,
\end{eqnarray}
with
\begin{equation}\label{eq4}
	\Psi(\omega;x) = \frac{1}{2\pi i}\int_{\mathcal{L}}\exp\bigg[\frac{z^3}{6}-\frac{x}{4}z^2-\omega z\bigg]dz\,,
\end{equation}
and $\mathcal{L}$ is any contour that starts at the point infinity with argument $\textnormal{arg}\ z=-\frac{\pi}{3}$ and ends at the conjugated point.  Notice that we have used slightly different variables than in \cite{Dot}.  In order to prove the relation \eqref{eq6}, we need to show that for any $t \in  \R$,
\begin{equation}\label{eq12}
	-W'(t) = \frac{1}{4}\int_{-\infty}^{\infty}\hat{P}(2^{-2/3}m,2^{-4/3}t)dm \,.
\end{equation}
From a quick examination of \eqref{in:70} and \eqref{eq1}, one sees that \eqref{eq12} holds given the following lemma, which is subsequently proven.

\begin{lemma}For any $x_1,x_2\geq 0, m\in\mathbb{R}$, and $t \in \R$,
\begin{equation}\label{eq13}
	-\frac{\partial}{\partial t}\Phi_{x_2x_1}(m,t) = 2^{-5/3}\psi(2^{1/3}x_1;-2^{-4/3}t,2^{-2/3}m)\psi(2^{1/3}x_2;2^{-4/3}t,2^{-2/3}m).
	\end{equation}
\end{lemma}

\begin{proof}
Consider first the right hand side of \eqref{eq13}.   Since
\begin{equation*}
	\textnormal{Ai}(z) = \frac{1}{2\pi i}\int_{\mathcal{L}}e^{\frac{1}{3}t^3-tz}dt,\hspace{0.5cm}z\in\mathbb{C}\,,
\end{equation*}
we have
\begin{equation*}
	\psi\big(2^{1/3}x;-2^{-4/3}t,2^{-2/3}m\big)=\frac{e^{-\frac{xt}{2}}}{\pi i}\int_{\mathcal{L}}\big(-2^{-4/3}t-\zeta\big)e^{\frac{\zeta^3}{3}-\zeta(2^{-8/3}t^2+2^{-2/3}m+2^{1/3}x)}d\zeta\,,
\end{equation*}
which after the change of variables $\zeta=2^{-1/3}\eta-2^{-4/3}t$, simplifies to
\begin{equation*}\label{eq14}
	\psi\big(2^{1/3}x;-2^{-4/3}t,2^{-2/3}m\big)=-\frac{2^{-2/3}}{\pi i}e^{\frac{t^3}{24}+\frac{tm}{4}}\int_{\mathcal{L}}\eta e^{\frac{\eta^3}{6}-\eta(x+\frac{m}{2})-\frac{t}{4}\eta^2}d\eta\,.
\end{equation*}
The right hand side of \eqref{eq13} can now be written as
\begin{equation}\label{eq144}
\begin{aligned}
	2^{-5/3}\psi(2^{1/3}&x_1;-2^{-4/3}t,2^{-2/3}m)\psi(2^{1/3}x_2;2^{-4/3}t,2^{-2/3}m) \\
	&=-\frac{1}{8\pi^2}\int_{\mathcal{L}}\int_{\mathcal{L}}\eta_1\eta_2e^{\frac{\eta_1^3+\eta_2^3}{6}-\eta_1(x_1+\frac{m}{2})
	-\eta_2(x_2+\frac{m}{2})-\frac{t}{4}(\eta_1^2-\eta_2^2)}d\eta_1d\eta_2\,.
	\end{aligned}
\end{equation}
For the left hand side of \eqref{eq4} we have
\begin{equation*}
\begin{aligned}
	\Psi(x_2+\frac{m}{2}+y;t)&\Psi(x_1+\frac{m}{2}+y;-t)\\
	&=-\frac{1}{4\pi^2}\int_{\mathcal{L}}\int_{\mathcal{L}}e^{\frac{\zeta_1^3+\zeta_2^3}{6}-\zeta_1(x_2+\frac{m}{2})-\zeta_2(x_1+\frac{m}{2})-\frac{t}{4}(\zeta_1^2-\zeta_2^2)
	-y(\zeta_1+\zeta_2)}d\zeta_1d\zeta_2\,,
	\end{aligned}
\end{equation*}
and
\begin{equation*}
\begin{aligned}
\Psi(x_2+\frac{m}{2}-y;t)&\Psi(x_1+\frac{m}{2}+y;-t)\\
	&=-\frac{1}{4\pi^2}\int_{\mathcal{L}}\int_{\mathcal{L}}e^{\frac{\zeta_1^3+\zeta_2^3}{6}-\zeta_1(x_2+\frac{m}{2})-\zeta_2(x_1+\frac{m}{2})-\frac{t}{4}(\zeta^2-\zeta_2^2)
	+y(\zeta_1-\zeta_2)}d\zeta_1d\zeta_2.
\end{aligned}
\end{equation*}
Plugging these formulas into \eqref{eq3} and interchanging differentiation and integration, we obtain
\begin{eqnarray}\label{eq15}
	\Phi_{x_2x_1}(m,t) &=&-\frac{1}{8\pi^2}\int_0^{\infty}\Bigg[\int_{\mathcal{L}}\int_{\mathcal{L}}e^{\frac{\zeta_1^3+\zeta_2^3}{6}-\zeta_1(x_2+\frac{m}{2})-\zeta_2(x_1+\frac{m}{2})-\frac{t}{4}(\zeta_1^2-\zeta_2^2)}\nonumber\\
	&&\times \bigg(e^{-y(\zeta_1+\zeta_2)}\big(\zeta_1-\zeta_2\big)+e^{y(\zeta_1-\zeta_2)}\big(\zeta_1+\zeta_2\big)\bigg)d\zeta_1d\zeta_2\Bigg]dy.
\end{eqnarray}
The function
\begin{equation*}\label{eq16}
	f(y) := e^{-y(\zeta_1+\zeta_2)}\big(\zeta_1-\zeta_2\big)+e^{y(\zeta_1-\zeta_2)}\big(\zeta_1+\zeta_2\big),\ \ \zeta_i\in\mathcal{L},\ \ \zeta_1\neq \zeta_2\,,
\end{equation*}
is integrable over $[0,\infty)$ since for $\zeta_i\in\mathcal{L}$ we have $\textnormal{Re}\ \zeta_i>0$, and if $\textnormal{Re}\ \zeta_1<\textnormal{Re}\ \zeta_2$, then
\begin{equation*}
	\big|f(y)\big|\leq C\ \textnormal{Re}(\zeta_2)e^{-y\textnormal{Re}(\zeta_2-\zeta_1)}.
\end{equation*}
On the other hand for $\textnormal{Re}\ \zeta_1>\textnormal{Re}\ \zeta_2$ we have for $\delta>0$
\begin{equation*}
	\big|f(y)\big|\leq C\ \textnormal{Re}(\zeta_1)\Big(e^{-2y\textnormal{Re}(\zeta_1)}+e^{-y\delta}\Big),
\end{equation*}
hence integrability follows.  Since for $\z_1 \ne \z_2$,
\begin{equation*}
	\int_0^{\infty}f(y)dy = \frac{\zeta_1-\zeta_2}{\zeta_1+\zeta_2}-\frac{\zeta_1+\zeta_2}{\zeta_1-\zeta_2} = -\frac{4\zeta_1\zeta_2}{\zeta_1^2-\zeta_2^2}\,,
\end{equation*}
we can use Fubini's theorem in \eqref{eq15} to obtain
\begin{equation*}\label{eq17}
	\Phi_{x_2x_1}(m,-t) = \frac{1}{2\pi^2}\int_{\mathcal{L}}\int_{\mathcal{L}}\frac{\zeta_1\zeta_2}{\zeta_1^2-\zeta_2^2}e^{\frac{\zeta_1^3+\zeta_2^3}{6}-\zeta_1(x_2+\frac{m}{2})-\zeta_2(x_1+\frac{m}{2})+\frac{t}{4}(\zeta_1^2-\zeta_2^2)}d\zeta_1d\zeta_2.
\end{equation*}
This integral appears singular, but one of the contours $\mathcal{L}$ could be deformed from the beginning such that $\z_1$ and $\z_2$ never coincide.  Now differentiating with respect to $t$ (and exchanging integration and differentiation) gives
\begin{equation*}
	\frac{\partial}{\partial t}\Phi_{x_2x_1}(m,-t) = \frac{1}{8\pi^2}\int_{\mathcal{L}}\int_{\mathcal{L}}\zeta_1\zeta_2\,e^{\frac{\zeta_1^3+\zeta_2^3}{6}-\zeta_1(x_2+\frac{m}{2})-\zeta_2(x_1+\frac{m}{2})+\frac{t}{4}(\zeta_1^2-\zeta_2^2)}d\zeta_1d\zeta_2.
\end{equation*}
Comparing the last line with \eqref{eq144}, we have shown
\begin{equation*}
	\frac{\partial}{\partial t}\Phi_{x_2x_1}(m,-t)= 2^{-5/3}\psi(2^{1/3}x_1;-2^{-4/3}t,2^{-2/3}m)\psi(2^{1/3}x_2;2^{-4/3}t,2^{-2/3}m).
\end{equation*}
The change of variables $t \mapsto -t$ gives \eqref{eq13}.
\end{proof}

\end{appendix}

\end{document}